\documentclass{article} 
\usepackage{iclr2027_conference,times}

\usepackage{amsmath,amsfonts,bm}

\def\eqref#1{equation~\ref{#1}}

\def\1{\bm{1}}

\DeclareMathAlphabet{\mathsfit}{\encodingdefault}{\sfdefault}{m}{sl}
\SetMathAlphabet{\mathsfit}{bold}{\encodingdefault}{\sfdefault}{bx}{n}

\usepackage{hyperref}
\usepackage{url}

\usepackage{microtype}
\usepackage{graphicx}
\usepackage[dvipsnames]{xcolor}

\usepackage{subcaption}
\usepackage{booktabs} 
\usepackage{multirow}
\usepackage{tabularx}
\usepackage{rotating}
\usepackage{wrapfig}

\usepackage{amsmath}
\usepackage{amssymb}
\usepackage{mathtools}
\usepackage{amsthm}

\usepackage[capitalize,noabbrev]{cleveref}

\usepackage{makecell}

\theoremstyle{plain}
\newtheorem{theorem}{Theorem}[section]
\newtheorem{proposition}[theorem]{Proposition}
\newtheorem{lemma}[theorem]{Lemma}

\theoremstyle{definition}

\newtheorem{assumption}[theorem]{Assumption}
\theoremstyle{remark}
\newtheorem{remark}[theorem]{Remark}

\title{Descent-Guided Policy Gradient for Scalable Cooperative Multi-Agent Learning}

\author{%
    Shan Yang$^{\sharp}$\thanks{Corresponding authors.}, \quad Yang Liu$^{\flat}$\footnotemark[\value{footnote}] \\
    National University of Singapore \\
    \texttt{$^{\sharp}$yang\_shan@u.nus.edu}, \quad \texttt{$^{\flat}$ceelya@nus.edu.sg}
}

\iclrfinalcopy 
\begin{document}

\maketitle
\lhead{Preprint} 

\begin{abstract}
    Scaling cooperative multi-agent reinforcement learning (MARL) is fundamentally limited by cross-agent noise. When agents share a common reward, each agent's learning signal is computed from a shared return that depends on all agents, so the stochasticity of the other agents enters the signal as cross-agent noise that grows with $N$. Many engineering systems, such as cloud computing and power systems, have differentiable analytical models that prescribe efficient system states, providing a new reference beyond noisy shared returns. In this work, we propose Descent-Guided Policy Gradient (DG-PG), a framework that augments policy-gradient updates with a low-noise descent signal toward this reference. We prove that DG-PG reduces policy-gradient estimator variance from $\mathcal{O}(N)$ to $\mathcal{O}(1)$, preserves the stationary points of the cooperative objective, and achieves agent-independent sample complexity $\mathcal{O}(1/\epsilon^3)$. We evaluate DG-PG on cloud resource scheduling and power dispatch tasks, which involve discrete and continuous decisions, respectively. In both tasks, DG-PG converges within 10 episodes on average at every scale up to 1500 heterogeneous agents.
\end{abstract}

\section{Introduction}
\label{sec:introduction}

Multi-agent reinforcement learning (MARL) has demonstrated remarkable success across a wide range of domains, from game playing \citep{yu2022surprising} to cloud resource scheduling \citep{mao2019learning} and transportation management \citep{el2013multiagent, yang2024markov}. In these applications, an agent represents an individual decision-making entity, such as a player in a game, a job dispatcher in cloud scheduling, or a vehicle in a road network. When agents share a common objective and must coordinate their actions, the setting is referred to as cooperative MARL. The training difficulty of cooperative MARL grows rapidly with the number of agents $N$, stemming from multi-agent credit assignment: when all agents jointly contribute to the shared objective, no individual agent's contribution can be cleanly identified.

Specifically, in cooperative MARL, agents update their policies with a learning signal estimated from the shared return, such as a policy gradient or an advantage. Regardless of its form, this learning signal carries cross-agent noise that grows with $N$. \citet{kuba2021settling} formalize this in the policy-gradient setting, showing that the noise contributes $\mathcal{O}(N)$ variance to the per-agent gradient. Thus, to make cooperative MARL scalable, this cross-agent noise must be addressed.

A common approach to alleviating cross-agent noise is credit decomposition, which extracts per-agent contributions from the shared return, thereby reducing the variance of the per-agent gradient estimator. This is typically realized by subtracting a counterfactual baseline \citep{foerster2018counterfactual}, factorizing the joint $Q$ into per-agent $Q$-functions \citep{rashid2020monotonic, wang2020qplex}, or estimating each agent's marginal contribution \citep{wolpert2001optimal}. However, in these methods, the cross-agent noise still grows with $N$, dominating the per-agent signal at scale. 

Fortunately, many cooperative engineering systems have analytical models that identify efficient system states, such as balanced server workloads derived from queueing models in cloud scheduling and cost-minimizing generation schedules derived from optimal power flow models in power systems. These efficient states describe desirable operation at the system level, but they do not indicate how individual agents should act. Can we then use these states as references to derive agent-level policy optimization directions that reduce the cross-agent noise in cooperative policy-gradient learning?

To this end, we propose Descent-Guided Policy Gradient (DG-PG), a framework that integrates such analytical models into the multi-agent policy gradient. Our key contributions are:
\begin{itemize}
    \item We derive low-noise per-agent descent guidance by differentiating the deviation from the model-prescribed reference with respect to each agent’s action. We augment the policy gradient with this guidance, accelerating learning without architectural changes to existing methods.
    \item We prove that: (i)~the augmented gradient does not alter the stationary points of the original cooperative objective (Theorem~\ref{thm:invariance}); (ii)~the variance of the per-agent gradient estimator reduces from $\mathcal{O}(N)$ to $\mathcal{O}(1)$, independent of the number of agents (Theorem~\ref{thm:variance}); and (iii)~the sample complexity to reach an $\epsilon$-optimal solution is $\mathcal{O}(1/\epsilon^3)$ (Theorem~\ref{thm:complexity}).
    \item We evaluate DG-PG on cloud scheduling (discrete) and power dispatch (continuous) with up to 1500 heterogeneous agents, demonstrating scale-invariant convergence and outperforming carefully tuned MAPPO and IPPO as well as centralized heuristics at scale.
\end{itemize}

\section{Related Work}
\label{sec:related_work}

Classical methods address the credit assignment problem in cooperative MARL, typically by extracting per-agent contributions from the shared return. Difference Rewards \citep{wolpert2001optimal} isolate each agent's marginal contribution via counterfactual simulation, but require access to the reward function or a simulator. Value decomposition methods, such as VDN \citep{sunehag2017value}, QMIX \citep{rashid2020monotonic}, and QPLEX \citep{wang2020qplex}, factorize the joint Q-function into per-agent value contributions under structural assumptions (additivity, monotonicity), but fail on tasks whose optimal joint value violates these assumptions. COMA \citep{foerster2018counterfactual} computes per-agent counterfactual advantages from a centralized critic, with the critic conditioned on the joint action of all agents. Crucially, these methods still extract per-agent signals from the shared return, leaving each signal carrying cross-agent noise that grows with $N$. DG-PG instead obtains a low-noise descent direction via an analytical reference model and augments each agent's policy gradient with it, enabling stable per-agent learning that does not degrade as $N$ grows.

Another line of work addresses scalability in cooperative MARL through structural simplifications, restricting either how agents update or how they interact. HAPPO \citep{kuba2021trust} restricts agent updates to a sequential order, so that each agent optimizes against fixed teammate policies and can obtain its marginal advantage. This is effective but sacrifices parallel execution, as each agent must wait for all preceding updates. Mean-Field MARL \citep{yang2018mean} replaces the $N$-agent interaction with a population-level approximation that simplifies each agent's optimization problem. This achieves scalability but requires agent homogeneity, precluding application to heterogeneous settings. DG-PG, in contrast, requires no structural change to how agents update or interact.

Other lines of work also incorporate prior knowledge into RL, but in forms different from DG-PG. First, Guided Policy Search \citep{levine2013guided} and Residual RL \citep{johannink2019residual} derive a base policy from domain knowledge and optimize from there. However, they still suffer from cross-agent noise at large $N$, so the policy cannot improve substantially beyond the base. Second, Physics-Informed Neural Networks \citep{raissi2019physics} encode prior knowledge as a constraint in the loss, rather than as a basis for learning or a direction for optimization. In contrast, DG-PG harnesses prior knowledge of the system to construct low-noise gradient guidance, significantly reducing noise in the gradient direction and accelerating learning without altering the stationary points of the cooperative objective. We refer the reader to Appendix~\ref{app:comparison_table} for a complete comparison.

\section{Preliminaries}
\label{sec:preliminary}

We consider a cooperative MARL setting, formalized as a \textit{Cooperative Markov Game} (equivalently, a Dec-POMDP) \citep{oliehoek2016concise}. The game is defined by the tuple $\mathcal{M} = \langle \mathcal{N}, \mathcal{S}, \mathcal{A}, P, r, \gamma \rangle$, where $\mathcal{N} = \{1, \dots, N\}$ denotes the set of $N$ agents, $\mathcal{S}$ is the global state space, and $\mathcal{A} = \times_{i \in \mathcal{N}} \mathcal{A}^i$ is the joint action space with $\mathcal{A}^i$ being agent $i$'s action space. At each timestep $t$, the system is in state $s_t \in \mathcal{S}$. Each agent $i$ receives a local observation $o_t^i$ (derived from $s_t$) and selects an action $a_t^i \in \mathcal{A}^i$ according to its policy $\pi_{\theta^i}(a_t^i | o_t^i)$, parameterized by $\theta^i$. The environment then transitions to state $s_{t+1}$ according to the dynamics $P(s_{t+1} | s_t, \boldsymbol{a}_t)$, where $\boldsymbol{a}_t = (a_t^1, \dots, a_t^N)$ denotes the joint action.

\subsection{Cooperative Objective and Policy Gradient}

In the cooperative setting, all agents share a common global reward function $r(s_t, \boldsymbol{a}_t)$. The shared objective is to maximize the expected discounted cumulative reward:
\begin{equation}
    J(\boldsymbol{\theta}) = \mathbb{E}_{\boldsymbol{\pi}} \left[ \sum_{t=0}^{\infty} \gamma^t r(s_t, \boldsymbol{a}_t) \right],
\end{equation}
where $\gamma \in [0, 1)$ is the discount factor and $\boldsymbol{\pi}$ is the joint policy. Standard policy gradient methods perform gradient ascent on $J(\boldsymbol{\theta})$ using the sampled return, $R_t = \sum_{k=0}^{\infty} \gamma^k r_{t+k}$. For agent $i$ with parameters $\theta^i$, the Policy Gradient Theorem \citep{sutton1999policy} gives the gradient:
\begin{equation}
    \nabla_{\theta^i} J(\boldsymbol{\theta}) = \mathbb{E}_{\boldsymbol{\pi}} \left[ \sum_{t=0}^{\infty} \gamma^t \nabla_{\theta^i} \log \pi_{\theta^i}(a_t^i | o_t^i) \, R_t \right].
    \label{eq:pg}
\end{equation}

\subsection{The Variance Explosion Problem}
\label{subsec:variance_explosion}
While Equation~\ref{eq:pg} provides an unbiased gradient estimator, its variance poses a critical challenge in large-scale settings. Because the estimator relies on sampled returns that depend on the joint action $\boldsymbol{a}_t$ of all $N$ agents, each agent's gradient estimate carries noise from every other agent's sampled actions. Following \citet{kuba2021settling}, this variance can be approximately decomposed as:
\begin{equation}
    \text{Var}(\hat{g}^i) \approx \underbrace{\sigma_{\text{self}}^2}_{\text{Own Exploration}} + \underbrace{(N-1) \sigma_{\text{others}}^2}_{\text{Cross-Agent Noise}}.
    \label{eq:variance_decomp}
\end{equation}
The second term represents the noise from sampling other agents' actions. As $N$ increases, this term dominates, causing the total estimation variance to scale as $\mathcal{O}(N)$. Section~\ref{subsec:alpha} verifies this $\mathcal{O}(N)$ growth in a controlled congestion game. This scaling implies that maintaining a constant estimation error requires the batch size to grow with $N$, rendering \textit{standard multi-agent policy gradient methods sample-inefficient for large-scale systems}.


\section{Method: Descent-Guided Policy Gradient}
\label{sec:method}

In this section, we present DG-PG, a framework that addresses variance explosion by exploiting the analytical models available in many engineering systems. We first define the \textit{reference state} as an efficient system state prescribed by analytical models and identify the conditions under which it yields a valid guidance term (Section~\ref{subsec:prior_assumptions}). We then augment the policy gradient with a low-noise, per-agent signal derived from the reference (Section~\ref{sec:descent_guided_formulation}). Under mild assumptions, DG-PG does not alter the stationary points of the cooperative objective while achieving agent-independent variance.

\subsection{Utilizing Analytical Priors: System State and Reference}
\label{subsec:prior_assumptions}

In many engineering systems, analytical models from operations research and control theory, such as fluid-limit approximations for queueing systems or equilibrium analysis for traffic networks, provide principled characterizations of desirable system behavior. While not executable under real-world stochasticity, these models offer something valuable, namely a \textit{reference} that describes what efficient system operation looks like under given conditions. Our key observation is that such a reference, if properly defined, can supply each agent with a clear, directed signal for policy improvement. This signal is far more informative than the noisy scalar reward $r_t$ aggregated over all agents.

As noted above, domain-specific models can prescribe efficient system states, but these states must be expressed in a concrete representational space. We define a \textit{system state} $\boldsymbol{x}_t \in \mathbb{R}^{K}$ as the vector of system-level quantities jointly determined by agents' actions (e.g., per-server resource utilization, per-generator power output, per-link flow), where $K$ is the dimension of the system state. The \textit{reference state} $\tilde{\boldsymbol{x}}_t \in \mathbb{R}^{K}$ is the efficient state that the analytical model yields under current conditions, computed independently of the learned policies. A key structural property of $\boldsymbol{x}_t$ is that each dimension $x_t^k$ is typically affected by the actions of only a limited number of agents, which bounds its one-step change independently of $N$, as formalized below. 

\begin{assumption}[Bounded State Increments]
    \label{ass:bounded_system_state}
    There exists a constant $C_x$, independent of the number of agents $N$, such that for any $k$, 
    \begin{equation}
        \sup_{t,s_t}
        \mathbb{E}_{\boldsymbol{a}_t\sim\boldsymbol{\pi}_{\boldsymbol{\theta}}(\cdot\mid s_t)}
        \left[
        \left(x_{t+1}^k-x_{t}^k\right)^2
        \mid s_t
        \right]
        \le C_x.
    \end{equation}
\end{assumption}

This property is crucial for deriving low-noise policy optimization signals (Section~\ref{sec:descent_guided_formulation}). In most engineering systems, we can identify such a system state $\boldsymbol{x}$. For example, each server in cloud scheduling receives only a limited number of new tasks at each step, while the flow on each link in routing changes only by a limited amount at each step. For the reference to yield a valid guidance signal, it should also ideally satisfy two conditions:

\begin{assumption}[Exogeneity]
    \label{ass:exogeneity}
    The reference state $\tilde{\boldsymbol{x}}_t$ does not depend on the policy parameters $\boldsymbol{\theta}$:
    \begin{equation}
        \nabla_{\boldsymbol{\theta}} \tilde{\boldsymbol{x}}_t = \mathbf{0}.
    \end{equation}
\end{assumption}

\begin{assumption}[Descent-Aligned Reference]
    \label{ass:alignment}
    Moving the system state toward the reference improves system performance. Formally, for any $\boldsymbol{x}_t \neq \tilde{\boldsymbol{x}}_t$:
    \begin{equation}
        \left\langle \nabla_{\boldsymbol{x}} J_{\boldsymbol{\theta}}(\boldsymbol{x}_t),\; \tilde{\boldsymbol{x}}_t - \boldsymbol{x}_t \right\rangle > 0.
    \end{equation}
\end{assumption}
Here, $J_{\boldsymbol{\theta}}(\boldsymbol{x}_t)$ denotes the expected cooperative return from $\boldsymbol{x}_t$ under the joint policy $\boldsymbol{\pi}_{\boldsymbol{\theta}}$.

When both conditions hold, the reference provides a stable, directionally correct target that does not shift as policies update (Assumption~\ref{ass:exogeneity}), and approaching it is expected to improve the cooperative objective (Assumption~\ref{ass:alignment}). This means that any policy update that steers the system state toward $\tilde{\boldsymbol{x}}_t$ is aligned with the original optimization goal, a property we exploit in the next subsection to construct a low-variance gradient estimator. Both assumptions are broadly satisfied in structured engineering systems, and we verify them concretely for the cloud scheduling case in Appendix~\ref{app:verification_cloud}.

Importantly, $\tilde{\boldsymbol{x}}_t$ need not be reachable, because it may arise from idealized fluid-limit or steady-state models. DG-PG uses $\tilde{\boldsymbol{x}}_t$ only to define a descent direction, not as a target to be attained exactly.

\subsection{Descent-Guided Policy Gradient}
\label{sec:descent_guided_formulation}

Given a reference $\tilde{\boldsymbol{x}}_t$ satisfying Assumptions~\ref{ass:exogeneity}--\ref{ass:alignment}, steering the system toward $\tilde{\boldsymbol{x}}_t$ at each step is directionally aligned with improving $J$. The challenge is to incorporate this information into the learning process in a way that (i) does not alter the original cooperative objective, and (ii) provides each agent with a low-noise gradient signal. Our DG-PG achieves both by augmenting the policy gradient with a per-agent guidance term derived from the analytical model.

We construct DG-PG in three steps: define a deviation measure between the system state and the reference, differentiate it to obtain per-agent descent directions, and augment the standard policy gradient with these directions to prevent cross-agent noise from scaling with $N$.

\textbf{Deviation Measure.}
We first define a per-step deviation function to quantify how far the system state deviates from the reference:
\begin{align}
    d(\boldsymbol{x}_t, \tilde{\boldsymbol{x}}_t) = \frac{1}{2} \| \boldsymbol{x}_t - \tilde{\boldsymbol{x}}_t \|^2.
    \label{eq:deviation}
\end{align}
By Assumption~\ref{ass:alignment}, reducing $d$ at each step is directionally consistent with improving $J$. \textit{However, simply incorporating $d$ into the reward would yield a scalar signal that cannot distinguish each agent's individual contribution to the deviation.} Because $d$ depends on each $a_t^i$ through the system state $x_t$, the derivative $\partial d / \partial a_t^i$ indicates whether agent $i$'s action $a_t^i$ moves the system toward or away from the reference $\tilde{\boldsymbol{x}}_t$, and can therefore guide agent $i$'s policy update.

\textbf{DG-PG Gradient.} We augment the policy gradient with this guidance, yielding the DG-PG gradient:
\begin{align}
    \nabla_{\theta^i} J_{\alpha}
     \triangleq (1-\alpha) \, \nabla_{\theta^i} J - \alpha \, \nabla_{\theta^i} \mathcal{G},
    \label{eq:augmented_gradient_def}
\end{align}
where $\alpha \in (0,1)$ is the guidance weight, and $\mathcal{G} \triangleq \mathbb{E}_{s \sim \nu^{\boldsymbol{\pi}}} \!\left[ d(\boldsymbol{x}, \tilde{\boldsymbol{x}}) \right]$ is the expected deviation under the state visitation distribution $\nu^{\boldsymbol{\pi}}$ induced by the joint policy. The DG-PG gradient $\nabla_{\theta^i} J_{\alpha}$ thus retains the optimization direction of $J$ while incorporating a guidance signal from the reference model. We prove that this augmentation preserves the stationary points of the original objective (Theorem~\ref{thm:invariance}).

To evaluate $\nabla_{\theta^i} \mathcal{G}$, we isolate each agent's contribution to the deviation via the chain rule:
\begin{align}
    \frac{\partial d}{\partial a_t^i} = \left\langle \nabla_{\boldsymbol{x}} d, \, \frac{\partial \boldsymbol{x}_{t+1}}{\partial a^i_t} \right\rangle = \langle \boldsymbol{x}_{t+1} - \tilde{\boldsymbol{x}}_{t+1}, \, \boldsymbol{z}_t^i \rangle,
    \label{eq:guidance_coefficient}
\end{align}
where $\boldsymbol{z}_t^i \triangleq \frac{\partial \boldsymbol{x}_{t+1}}{\partial a^i_t}$ is agent $i$'s \textit{local influence vector}, the partial derivative of the system state with respect to its action. This coefficient measures the deviation projected onto the dimensions that agent $i$ influences. The randomness of this coefficient comes only from the one-step changes of these dimensions, which are bounded by Assumption~\ref{ass:bounded_system_state}, so its variance does not grow with $N$. The concrete form of $\boldsymbol{z}_t^i$ depends on the domain, and we derive it for cloud scheduling in Appendix~\ref{app:influence_derivation}.

Substituting Equation~\ref{eq:guidance_coefficient} into the policy gradient theorem yields the gradient of the guidance term:
\begin{align}
    \nabla_{\theta^i} \mathcal{G}
     = \mathbb{E}_{s \sim \nu^{\boldsymbol{\pi}},\, \boldsymbol{a} \sim \boldsymbol{\pi}} \left[ \langle \boldsymbol{x}_{t+1} - \tilde{\boldsymbol{x}}_{t+1}, \, \boldsymbol{z}_t^i \rangle \, \nabla_{\theta^i} \log \pi^i(a^i_t|o^i_t) \right].
    \label{eq:gradient_decomposition}
\end{align}

\textbf{Gradient Estimator.}
In practice, we estimate the DG-PG gradient $\nabla_{\theta^i} J_{\alpha}$ via the score function method. At each step $t$ of a sampled trajectory, the gradient estimate for agent $i$ is:
\begin{align}
    \hat{g}_{\text{DG},t}^i = (1-\alpha) \, \hat{g}_{J,t}^i - \alpha \, \hat{g}_{\mathcal{G},t}^i,
    \label{eq:dg_estimator}
\end{align}
where
\begin{align}
    \hat{g}_{J,t}^i &= R_t \, \nabla_{\theta^i} \log \pi^i(a^i_t|o^i_t), \\
    \hat{g}_{\mathcal{G},t}^i &= \langle \boldsymbol{x}_{t+1} - \tilde{\boldsymbol{x}}_{t+1}, \, \boldsymbol{z}^i_t \rangle \, \nabla_{\theta^i} \log \pi^i(a^i_t|o^i_t).
\end{align}
Each agent then updates its policy parameters via gradient ascent with learning rate $\eta > 0$:
\begin{equation}
    \theta^i_{k+1} = \theta^i_k + \eta \, \hat{g}_{\text{DG},t}^i.
    \label{eq:dg_update}
\end{equation}

\textbf{Variance Properties.}
The two components of $\hat{g}_{\text{DG},t}^i$ have fundamentally different variance properties. The standard policy-gradient estimate $\hat{g}_{J,t}^i$ uses the return $R_t$, which aggregates all $N$ agents' stochastic actions, yielding $\mathcal{O}(N)$ variance. The guidance estimate $\hat{g}_{\mathcal{G},t}^i$ uses an analytically computed coefficient $\langle \boldsymbol{x}_{t+1} - \tilde{\boldsymbol{x}}_{t+1}, \boldsymbol{z}^i_t \rangle$, contributing $\mathcal{O}(1)$ variance independent of $N$ under Assumption~\ref{ass:bounded_system_state}. We formalize how the choice of $\alpha$ controls this variance trade-off in Section~\ref{sec:theory}.

\subsection{Implementation}

\textbf{Integration with MAPPO.}
DG-PG can be integrated into the MAPPO framework \citep{yu2022surprising} by modifying the advantage estimator. After collecting trajectories and computing the Generalized Advantage Estimation (GAE) advantages $\hat{A}_{GAE}^i(t)$, they are augmented with the guidance term:
\begin{equation}
    \hat{A}_{DG}^i(t) = (1-\alpha) \hat{A}_{GAE}^i(t) - \alpha \cdot \langle \boldsymbol{x}_{t+1} - \tilde{\boldsymbol{x}}_{t+1}, \boldsymbol{z}_t^i \rangle.
    \label{eq:a_dg}
\end{equation}
The computational overhead is minimal, since the reference trajectory $\{\tilde{\boldsymbol{x}}_t\}$ is computed once per batch using the analytical model, and each guidance coefficient is just an inner product. The policy and critic are then updated using standard PPO objectives with the modified advantages. 

\textbf{Comparison with Existing Methods.}
Unlike counterfactual methods (e.g., COMA \citep{foerster2018counterfactual}) that require $N$ separate critics, or value decomposition methods (e.g., QMIX \citep{rashid2020monotonic}) that need additional mixing networks, DG-PG requires \textit{no architectural changes} to existing methods. The only modification is a single-line change to the advantage estimator in Equation~\ref{eq:a_dg}.

\section{Theoretical Guarantees}
\label{sec:theory}

In this section, we establish the theoretical foundations of DG-PG. Our central result is that the DG-PG gradient estimator achieves \textit{agent-independent variance} (Theorem~\ref{thm:variance}), breaking the $\mathcal{O}(N)$ scaling barrier of standard policy gradients. We complement this with two supporting guarantees: (1) \textit{consistency} (Theorem~\ref{thm:invariance}), ensuring that DG-PG preserves the stationary points of the original cooperative objective; and (2) \textit{sample complexity} (Theorem~\ref{thm:complexity}), showing that the variance reduction yields agent-independent convergence rates under standard regularity assumptions (the complete dependency structure is given in Table~\ref{tab:assumptions_deps}, Appendix~\ref{app:assumptions_table}).

\subsection{Preservation of Stationary Points}

A primary concern when augmenting the policy gradient is whether the additional guidance term shifts the stationary points of the original cooperative objective. We address this by showing that DG-PG preserves these stationary points.

\begin{theorem}[Preservation of Stationary Points]
    \label{thm:invariance}
    Let $\boldsymbol{\theta}^*$ be a joint policy parameter. If Assumptions~\ref{ass:exogeneity}--\ref{ass:alignment} hold and the policies are sufficiently expressive, then $\boldsymbol{\theta}^*$ is a stationary point of the original cooperative objective if and only if it is a stationary point of the DG-PG objective:
    \begin{equation}
        \nabla_{\theta^i} J(\boldsymbol{\theta}^*) = \mathbf{0}
        \ \forall i \in \mathcal{N}
        \quad \Longleftrightarrow \quad
        \nabla_{\theta^i} J_{\alpha}(\boldsymbol{\theta}^*) = \mathbf{0}
        \ \forall i \in \mathcal{N},
        \qquad
        \forall \alpha \in (0, 1).
    \end{equation}
\end{theorem}


\begin{proof}[Proof Sketch]
We first show that $\nabla_{\theta^i} \mathcal{G}(\boldsymbol{\theta}^*) = \mathbf{0}$ at any stationary point $\boldsymbol{\theta}^*$ (Lemma~\ref{lemma:guidance_vanishing}). Since both $\nabla_{\theta^i} J$ and $\nabla_{\theta^i} \mathcal{G}$ vanish at $\boldsymbol{\theta}^*$, the DG-PG gradient $\nabla_{\theta^i} J_{\alpha} = (1-\alpha) \cdot \mathbf{0} - \alpha \cdot \mathbf{0} = \mathbf{0}$. Conversely, Assumption~\ref{ass:alignment} ensures that the guidance direction $-\nabla_{\boldsymbol{\theta}}\mathcal{G}$ never opposes $\nabla_{\boldsymbol{\theta}} J$, so the DG-PG gradient can vanish only when $\nabla_{\boldsymbol{\theta}} J$ does. The detailed proof is in Appendix~\ref{app:proof_invariance}.
\end{proof}

\subsection{Variance Reduction}
\label{sec:variance_analysis}

Having established that DG-PG preserves the stationary points of $J$ (Theorem~\ref{thm:invariance}), we now formalize the variance properties of the DG-PG gradient estimator. The return $R_t$ aggregates all $N$ agents' stochastic actions, so variance of each agent's gradient estimate is $\mathrm{Var}(\hat{g}_J^i) = \mathcal{O}(N)$. In contrast, the guidance coefficient $\langle \boldsymbol{x}_{t+1} - \tilde{\boldsymbol{x}}_{t+1}, \boldsymbol{z}_t^i \rangle$ depends only on the dimensions that agent $i$ influences, so $\mathrm{Var}(\hat{g}_{\mathcal{G}}^i) = \mathcal{O}(1)$ under Assumption~\ref{ass:bounded_system_state} (Lemma~\ref{lemma:guidance_variance}). Section~\ref{subsec:alpha} validates both results empirically.

Since the DG-PG estimator combines $\hat{g}_J^i$ and $\hat{g}_{\mathcal{G}}^i$, its variance depends on their correlation $\rho = \mathrm{Corr}(\hat{g}_J^i, \hat{g}_{\mathcal{G}}^i)$. Under Assumption~\ref{ass:alignment}, improving $J$ reduces the deviation $\mathcal{G}$, so $\hat{g}_J^i$ and $\hat{g}_{\mathcal{G}}^i$ are negatively correlated ($\rho < 0$). The following theorem characterizes the optimal variance as a function of $\rho$ and the variances of two gradient estimators.

\begin{theorem}[Agent-Independent Variance]
    \label{thm:variance}
    Under Assumptions~\ref{ass:bounded_system_state}--\ref{ass:alignment}, let $\sigma_J^2 \triangleq \mathrm{Var}(\hat{g}_J^i)$, $\sigma_{\mathcal{G}}^2 \triangleq \mathrm{Var}(\hat{g}_{\mathcal{G}}^i)$, and $\rho \triangleq \mathrm{Corr}(\hat{g}_J^i, \hat{g}_{\mathcal{G}}^i) \in (-1, 0)$. The minimum variance over the guidance weight $\alpha$ is
    \begin{equation}
        \min_{\alpha \in (0,1)} \mathrm{Var}(\hat{g}_{\mathrm{DG}}^i) = \frac{\sigma_J^2 \, \sigma_{\mathcal{G}}^2 \,(1 - \rho^2)}{\sigma_J^2 + \sigma_{\mathcal{G}}^2 + 2\rho\,\sigma_J \sigma_{\mathcal{G}}}.
    \end{equation}
    Since $\sigma_J^2 = \mathcal{O}(N)$ and $\sigma_{\mathcal{G}}^2 = \mathcal{O}(1)$, this minimum variance is $\mathcal{O}(1)$, independent of $N$.
\end{theorem}

\begin{proof}[Proof Sketch]
    The variance of the DG-PG estimator $\sigma_{\mathrm{DG}}^2(\alpha)$ is a convex quadratic in $\alpha$ with cross-term $\rho\,\sigma_J\sigma_{\mathcal{G}}$. Minimizing in closed form yields 
    $$\sigma_{\mathrm{DG},\min}^2 = \sigma_J^2\sigma_{\mathcal{G}}^2(1-\rho^2)\,/\,({\sigma_J^2 + \sigma_{\mathcal{G}}^2 + 2\rho\,\sigma_J\sigma_{\mathcal{G}}}).$$ 
    Substituting $\sigma_J^2 = \mathcal{O}(N)$ and $\sigma_{\mathcal{G}}^2 = \mathcal{O}(1)$, the denominator is dominated by $\sigma_J^2$ and is thus $\mathcal{O}(N)$, while the numerator is $\mathcal{O}(N)$, giving $\sigma_{\mathrm{DG},\min}^2 = \mathcal{O}(1)$ for all $N$. 
\end{proof}

\textbf{The role of $\alpha$.} Appendix~\ref{app:proof_variance_complexity} shows that the optimal weight $\alpha^*$ is a function of $\sigma_J^2$, $\sigma_{\mathcal{G}}^2$, and $\rho$. Since only $\sigma_J^2$ grows with $N$, this optimizer shifts toward the guidance signal and satisfies $\alpha^*\to 1$ in the large-$N$ limit. At finite scales, $\alpha$ remains a tuning choice that balances variance reduction with the original policy-gradient signal, and we validate this choice empirically in Section~\ref{sec:experiments}.


\subsection{Convergence Rate}

The variance reduction in Theorem~\ref{thm:variance}, combined with standard regularity assumptions, yields agent-independent sample complexity.  We evaluate the scalability in Section~\ref{subsec:scalability}.

\begin{theorem}[Sample Complexity]
    \label{thm:complexity}
    Assume $J$ is $L$-smooth and satisfies the gradient domination condition with constant $\kappa_{\mathrm{GD}} > 0$ (see Assumption~\ref{ass:gradient_domination} in Appendix~\ref{app:proof_complexity}). Given the optimal variance from Theorem~\ref{thm:variance} and a fixed guidance weight $\alpha \in (0,1)$, DG-PG achieves an $\epsilon$-optimal policy in:
    \begin{equation}
        T = \mathcal{O}\left( \frac{L \kappa_{\mathrm{GD}}^4}{\epsilon^3} \right)
    \end{equation}
    iterations. This bound is independent of the number of agents $N$.
\end{theorem}


\begin{proof}[Proof Sketch]
$L$-smoothness bounds the per-iteration descent in terms of $\|\nabla J\|^2$ and the gradient variance $\sigma_{\mathrm{DG}}^2$. The gradient domination condition lower-bounds $\|\nabla J\|^2$ by the squared optimality gap $(J^* - J(\boldsymbol{\theta}_k))^2$, which yields a quadratic recursion for the gap. Substituting $\sigma_{\mathrm{DG}}^2 = \mathcal{O}(1)$ from Theorem~\ref{thm:variance} yields the final bound. The derivation is provided in Appendix~\ref{app:proof_complexity}.
\end{proof}

\textbf{Summary.} Theorems~\ref{thm:invariance}--\ref{thm:complexity} establish that DG-PG is (1) scalable: variance does not grow with $N$; (2) safe: stationary points of $J$ are preserved; and (3) efficient: sample complexity is agent-independent.

\section{Experiments}
\label{sec:experiments}


\subsection{Experimental Setup}
\label{subsec:setup}

Our experiments use three environments with complementary roles. A one-step congestion game provides a controlled check of the variance analysis in Theorem~\ref{thm:variance}. Cloud scheduling and power dispatch serve as realistic large-scale engineering tasks, with discrete and continuous actions respectively, on which we evaluate the scalability of DG-PG up to $N{=}1500$ agents.

\textbf{One-Step Congestion Game.} In this game, each of $N$ agents assigns one unit of demand to one of $N$ resources with heterogeneous capacities, and all agents share a quadratic congestion cost that penalizes overloaded low-capacity resources most heavily. This game offers a simple setting in which the gradient variance can be measured exactly. See Appendix~\ref{app:env_congestion} for the full setup.

\textbf{Cloud Scheduling.} This task requires $N$ job dispatcher agents to place arriving jobs on servers so as to reduce queueing delay and improve energy efficiency. We build the environment with $N$ servers based on AWS instance types, and it combines \textit{heterogeneous server efficiency, heavy-tailed bimodal workloads, and non-stationary arrivals}. Together, these factors create model mismatch by violating the stationarity assumptions behind the queueing-theoretic prior. Details are in Appendix~\ref{app:env_spec}.

\textbf{Power Dispatch.} This task requires $N$ thermal generators to adjust their power commands to minimize power imbalance and generation cost. The imbalance determines how much total power the generators must supply, while the cost depends on how this total is allocated among them. \textit{The generators are heterogeneous \citep{barrows2020ieee}, and the net load fluctuates with stochastic wind \citep{rozada2021deep}.} We adopt the economic dispatch solution as the reference, while ramp limits and generator dynamics prevent the system from reaching it instantly. Details are in Appendix~\ref{app:env_power}.

\subsection{Gradient Variance and Guidance Weight}
\label{subsec:alpha}

In this part, we verify the agent-independent variance of the guidance gradient in Theorem~\ref{thm:variance} using the one-step congestion game. Unlike more complex environments, this game admits the exact policy gradient needed to measure the variance. We then test the sensitivity of DG-PG to the guidance weight $\alpha$ in cloud scheduling. These results support our choice of $\alpha$ in the subsequent experiments.

\begin{figure}[!ht]
    \centering
    \includegraphics[width=\linewidth]{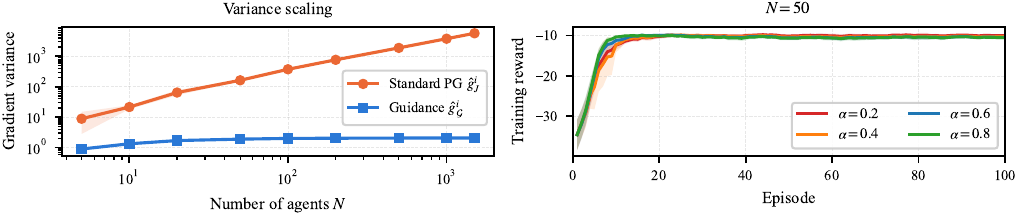}
    \caption{(a) Gradient variance scaling. (b) Guidance weight sensitivity.}
    \label{fig:alpha_sensitivity}
\end{figure}

\paragraph{Variance Scaling.}
Figure~\ref{fig:alpha_sensitivity}(a) shows that the variance $\sigma_J^2$ of the standard policy-gradient estimator, measured around the exact policy gradient (Appendix~\ref{app:variance_toy}), grows almost linearly with $N$, from $8.9$ at $N{=}5$ to $5.7\times10^{3}$ at $N{=}1500$. In contrast, the variance $\sigma_{\mathcal{G}}^2$ of the guidance estimator levels off at about $2$ from $N{=}100$ onward, consistent with Section~\ref{subsec:variance_explosion} and Theorem~\ref{thm:variance}. This widening gap between $\sigma_J^2$ and $\sigma_{\mathcal{G}}^2$ implies that larger systems should place more weight on the guidance estimator. 

\paragraph{Guidance Weights.}
Figure~\ref{fig:alpha_sensitivity}(b) shows training curves in cloud scheduling at $N{=}50$ for all tested $\alpha$ values. The results reveal a trade-off in which \textit{larger $\alpha$ accelerates early convergence, while smaller $\alpha$ achieves marginally superior best performance}. At $N{=}50$, $\alpha{=}0.2$ achieves $-8.75 \pm 0.21$, slightly outperforming $\alpha{=}0.8$ at $-9.04 \pm 0.12$. The small spread across the tested $\alpha$ values further indicates that DG-PG is \textit{robust to the choice of $\alpha$}.

\paragraph{Selected Hyperparameters.}
Based on the observed trade-off, we use a dynamic $\alpha$ schedule for scales within $N \leq 50$, starting from $\alpha{=}0.9$ for fast early convergence and decaying to $\alpha{=}0.2$ for later fine-tuning. For scales with $N \geq 100$, where cross-agent noise is more severe, we fix $\alpha{=}0.9$ to maintain strong guidance throughout training. Specific settings are provided in Appendix~\ref{app:hyperparams}.

\subsection{Controlled Comparison}
\label{subsec:main_results}

\textbf{Baselines.}
We compare DG-PG in cloud scheduling against (1)~\textbf{MAPPO}~\citep{yu2022surprising}, multi-agent PPO with a centralized value function, and (2)~\textbf{IPPO}, independent PPO with local value functions. We additionally report (3)~\textbf{Best-Fit}, a heuristic algorithm that operates with a centralized dispatcher, as an upper reference. Mean-Field methods \citep{yang2018mean} are excluded because they assume agent homogeneity. DG-PG, MAPPO, and IPPO are trained under matched core hyperparameters (Appendix~\ref{app:hyperparams}). 

\textbf{Results.}
We evaluate $N \in \left\{2, 5, 10\right\}$ in the cloud scheduling environment, and train MAPPO and IPPO for 500 episodes, $2.5\times$ longer than DG-PG (Figure~\ref{fig:learning_curves}). At the best checkpoint of each run, all methods perform comparably at $N{=}2$. At $N{=}5$, DG-PG reaches $-10.97 \pm 0.38$, while MAPPO and IPPO reach $-21.63 \pm 0.77$ and $-23.04 \pm 0.29$. At $N{=}10$, the gap widens, with DG-PG reaching $-9.69 \pm 0.23$ compared with $-29.40 \pm 2.07$ for MAPPO and $-28.71 \pm 2.81$ for IPPO. These results show that the gap widens as $N$ increases, consistent with the cross-agent noise bottleneck.

\begin{figure}[!ht]
    \centering
    \includegraphics[width=\linewidth]{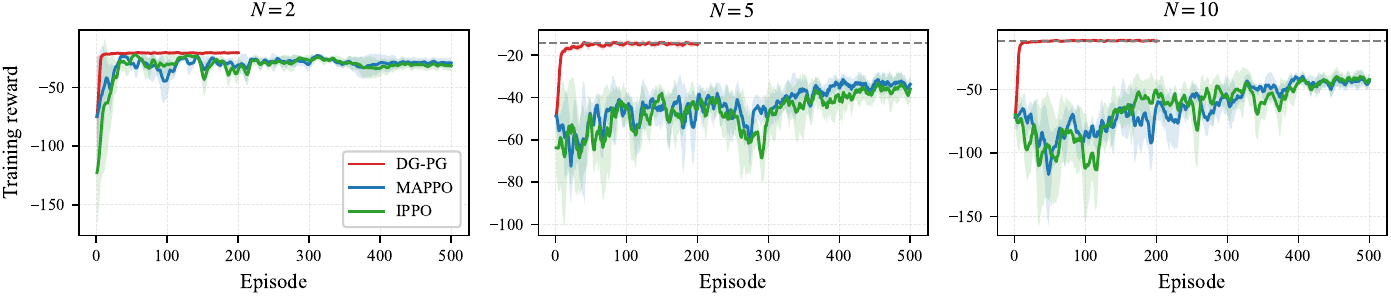}
    \caption{Training reward under the controlled comparison at $N \in \{2, 5, 10\}$.}
    \label{fig:learning_curves}
\end{figure}

\subsection{Scalability}
\label{subsec:scalability}

We next scale DG-PG to $N{=}1500$ agents in cloud scheduling and power dispatch. MAPPO and IPPO are excluded because they fail to converge at these scales (Appendix~\ref{app:nn_mappo_failure}).

\paragraph{Scale-Invariant Convergence.}
Figure~\ref{fig:scalability_combined}(a) and (c) show DG-PG's large-scale training curves up to $N{=}1500$, with scales below $N{=}50$ omitted in (a) and (c) for visual clarity. In both environments, DG-PG reaches within 10\% of its converged reward in fewer than 10 episodes on average (Appendix~\ref{app:compute}), and the required number of episodes remains bounded as $N$ increases. This bounded convergence is consistent with the sample complexity in Theorem~\ref{thm:complexity}. 

\paragraph{Best Checkpoint Performance.}
Figure~\ref{fig:scalability_combined}(b), (d) and Table~\ref{tab:main_results} compare DG-PG at its best checkpoint with Best-Fit and Random Dispatch in cloud scheduling and with a centralized greedy controller and automatic generation control (AGC) in power dispatch (Appendices~\ref{app:env_spec} and~\ref{app:env_power}). Although Best-Fit assigns each task using the global state of all servers, DG-PG matches Best-Fit at smaller scales and surpasses it from $N{=}200$ onward, and significantly outperforms Random Dispatch at all scales. In power dispatch, DG-PG also outperforms both centralized greedy and AGC, the standard scheme for frequency regulation in power system operation, at all scales. \textit{This makes DG-PG suitable for distributed systems, where no centralized controller needs to compute joint actions online, and each agent can act directly through its learned policy.}

\paragraph{Computational Cost.}
DG-PG is also cheap to train. It requires only small neural networks and no complex training architecture, which allows it to run on modest CPU-only hardware. For instance, all cloud scheduling experiments were run on a personal laptop (Apple M1, 16GB). Even so, owing to its fast convergence and simple training, DG-PG needs only about 35 minutes of training at $N{=}1500$, mostly spent on environment simulation. Appendix~\ref{app:compute} reports the timing at each scale.

\begin{figure}[!ht]
    \centering
    \includegraphics[width=\linewidth]{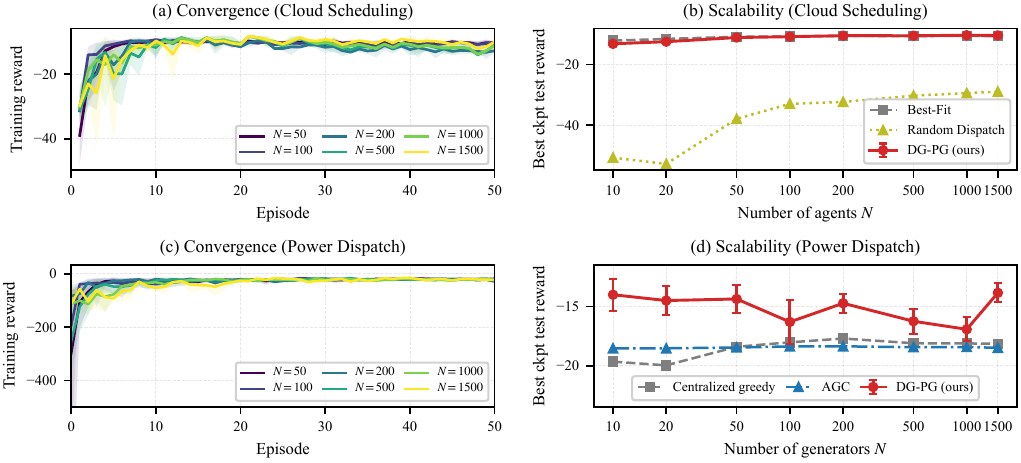}
    \caption{Scalability and convergence.}
    \label{fig:scalability_combined}
\end{figure}

\begin{table}[!ht]
\centering
\caption{Test reward across scales. DG-PG reports mean $\pm$ std over training seeds.}
\label{tab:main_results}
\resizebox{\columnwidth}{!}{%
\setlength{\tabcolsep}{2.5pt}
\begin{tabular}{lcccccccc}
\toprule
\textit{Cloud scheduling} & $N{=}10$ & $N{=}20$ & $N{=}50$ & $N{=}100$ & $N{=}200$ & $N{=}500$ & $N{=}1000$ & $N{=}1500$ \\
\midrule
DG-PG (ours) & $-13.19{\pm}0.14$ & $-12.49{\pm}0.10$ & $-11.13{\pm}0.03$ & $-10.84{\pm}0.02$ & $\mathbf{-10.51{\pm}0.11}$ & $\mathbf{-10.52{\pm}0.06}$ & $\mathbf{-10.36{\pm}0.08}$ & $\mathbf{-10.41{\pm}0.02}$ \\
Best-Fit & $\mathbf{-12.15}$ & $\mathbf{-11.58}$ & $\mathbf{-10.88}$ & $\mathbf{-10.77}$ & $-10.58$ & $-10.66$ & $-10.55$ & $-10.66$ \\
Random Dispatch & $-50.82$ & $-52.85$ & $-37.98$ & $-33.04$ & $-32.40$ & $-30.32$ & $-29.49$ & $-29.01$ \\
\midrule
\textit{Power dispatch} & $N{=}10$ & $N{=}20$ & $N{=}50$ & $N{=}100$ & $N{=}200$ & $N{=}500$ & $N{=}1000$ & $N{=}1500$ \\
\midrule
DG-PG (ours) & $\mathbf{-13.99{\pm}1.34}$ & $\mathbf{-14.49{\pm}1.21}$ & $\mathbf{-14.35{\pm}1.20}$ & $\mathbf{-16.30{\pm}1.85}$ & $\mathbf{-14.72{\pm}0.79}$ & $\mathbf{-16.24{\pm}1.05}$ & $\mathbf{-16.92{\pm}1.01}$ & $\mathbf{-13.83{\pm}0.81}$ \\
Centralized greedy & $-19.66$ & $-19.97$ & $-18.40$ & $-18.02$ & $-17.70$ & $-18.10$ & $-18.11$ & $-18.15$ \\
AGC & $-18.53$ & $-18.53$ & $-18.46$ & $-18.37$ & $-18.36$ & $-18.43$ & $-18.42$ & $-18.49$ \\
\bottomrule
\end{tabular}%
}
\end{table}

\textit{These experiments further show that DG-PG learns a standalone decentralized policy, rather than a policy that depends on the analytical reference at execution time.} The analytical model is used only to construct gradient guidance during training. Once training is complete, each agent acts through its learned policy without guidance or an online centralized scheduler. This differs from residual policy methods, where the learned component remains tied to a base controller at deployment. 

\section{Conclusion}

In this work, we identified cross-agent noise as the central bottleneck in scaling cooperative MARL and proposed DG-PG, which guides policy-gradient updates with a per-agent descent signal derived from an analytical reference. This guidance reduces the variance of each agent's gradient estimate from $\mathcal{O}(N)$ to $\mathcal{O}(1)$ and yields agent-independent sample complexity while preserving the stationary points of the cooperative objective. On cloud scheduling and power dispatch with up to 1500 heterogeneous agents, DG-PG converges within 10 episodes on average at every scale.

DG-PG shows that analytical models can be used not merely to provide heuristics or controllers, but to address cross-agent noise in cooperative MARL. It does not require an exact optimum, an attainable target, or an executable policy, only a directionally aligned reference. DG-PG offers a principled route to reliable distributed decision-making and efficient learning in large-scale engineered systems. 

\textbf{Limitations and future work.} DG-PG requires a reference that is directionally aligned with the cooperative objective. Domains without such structure need alternative guidance, and applying DG-PG to a new system requires specifying the system state, deriving the reference, and verifying its alignment. Future work can extend DG-PG to other engineered systems, such as traffic networks.

\bibliography{bibliography}
\bibliographystyle{iclr2027_conference}

\newpage
\appendix

\section{Comparison with Existing Methods}
\label{app:comparison_table}

Table~\ref{tab:method_comparison} provides a systematic comparison of DG-PG with representative methods from major paradigms in cooperative MARL. We evaluate methods across key dimensions, including gradient variance scaling, sample complexity, support for heterogeneous policies, parallel updates, architectural overhead, and domain knowledge requirements. The comparison highlights DG-PG's unique position as the only method that simultaneously achieves $\mathcal{O}(1)$ gradient variance, supports heterogeneous agent policies with parallel updates, and requires no additional learned components beyond a standard centralized critic.

\begin{sidewaystable}
    \centering
    \caption{Systematic comparison of cooperative MARL methods across key architectural and algorithmic dimensions. DG-PG is the only method that simultaneously achieves $\mathcal{O}(1)$ variance, $\mathcal{O}(1/\epsilon^3)$ sample complexity, full heterogeneity support, parallel updates, and zero extra learned components.}
    \label{tab:method_comparison}
    \footnotesize
    \begin{tabular}{@{}llcccccl@{}}
        \toprule
        \textbf{Method} & \textbf{Approach} & \makecell{\textbf{Variance}\\\textbf{Scaling}} & \makecell{\textbf{Sample}\\\textbf{Complexity}} & \makecell{\textbf{Heterog.}\\\textbf{Policies}} & \makecell{\textbf{Parallel}\\\textbf{Updates}} & \makecell{\textbf{Extra Learned}\\\textbf{Components}} & \makecell{\textbf{Assumption /}\\\textbf{Requirement}} \\
        \midrule
        \multicolumn{8}{l}{\textit{Policy Gradient Methods}} \\
        \quad MAPPO~\citep{yu2022surprising} & Centralized PG & $\mathcal{O}(N)$ & $\mathcal{O}(N/\epsilon^3)$ & \checkmark & \checkmark & None & None \\
        \quad IPPO~\citep{yu2022surprising} & Independent PG & $\mathcal{O}(N)$ & $\mathcal{O}(N/\epsilon^3)$ & \checkmark & \checkmark & None & None \\
        \quad HAPPO~\citep{kuba2021trust} & Sequential PG & $\mathcal{O}(1)^{\dagger}$ & $\mathcal{O}(N/\epsilon^3)^{\dagger}$ & \checkmark & \texttimes & None & None \\
        \quad COMA~\citep{foerster2018counterfactual} & Counterfactual & $\mathcal{O}(N)^{\ddagger}$ & $\mathcal{O}(N/\epsilon^3)$ & \checkmark & \checkmark & $N$ critic heads & None \\
        \midrule
        \multicolumn{8}{l}{\textit{Value Decomposition Methods}} \\
        \quad QMIX~\citep{rashid2020monotonic} & Mixing Network & N/A (Value) & --- & \checkmark & \checkmark & Mixing net & Monotonicity \\
        \quad VDN~\citep{sunehag2017value} & Additive & N/A (Value) & --- & \checkmark & \checkmark & None & Additivity \\
        \quad QPLEX~\citep{wang2020qplex} & Attention & N/A (Value) & --- & \checkmark & \checkmark & Attn mixer, $\mathcal{O}(N^2)$ & Dueling struct. \\
        \midrule
        \multicolumn{8}{l}{\textit{Mean-Field Approximation}} \\
        \quad Mean-Field~\citep{yang2018mean} & Neighbor Avg. & $\mathcal{O}(1)^{\S}$ & $\mathcal{O}(1/\epsilon^3)^{\S}$ & \texttimes & \checkmark & None & Exchangeability \\
        \midrule
        \multicolumn{8}{l}{\textit{Reward Modification Methods}} \\
        \quad PBRS~\citep{ng1999policy} & Potential Shaping & $\mathcal{O}(N)$ & $\mathcal{O}(N/\epsilon^3)$ & \checkmark & \checkmark & None & Potential func. \\
        \quad Diff.~Rewards~\citep{wolpert2001optimal} & Counterfactual & Problem-dep.$^{*}$ & Problem-dep.$^{*}$ & \checkmark & \checkmark & None & Exact simulator \\
        \midrule
        \multicolumn{8}{l}{\textit{\textbf{Gradient Guidance (Proposed)}}} \\
        \quad \textbf{DG-PG (Ours)} & \textbf{Analytic Guide} & $\boldsymbol{\mathcal{O}(1)}$ & $\boldsymbol{\mathcal{O}(1/\epsilon^3)}$ & \checkmark & \checkmark & \textbf{None} & \textbf{Analytical model} \\
        \bottomrule
    \end{tabular}

    \raggedright
    \footnotesize
    \textbf{Notes:} ``Variance Scaling'' = per-agent gradient variance w.r.t.\ $N$; ``Sample Complexity'' = number of iterations to reach $\epsilon$-optimality (proven bounds for PG methods); ``Extra Learned Components'' = trainable modules beyond standard actor-critic; ``---'' = not directly comparable (value-based methods lack explicit PG complexity bounds). \\[0.2em]
    $^{\dagger}$HAPPO reduces per-agent variance via sequential updates with trust-region constraints, but sacrifices parallelism. Consequently, wall-clock time per iteration scales as $\mathcal{O}(N)$, yielding effective sample complexity $\mathcal{O}(N/\epsilon^3)$ in wall-clock time. \\
    $^{\ddagger}$COMA requires $N$ separate critic heads to compute counterfactual baselines. Reduces own-action variance $\sigma_{\text{self}}^2$ but does not suppress cross-agent noise $(N-1)\sigma_{\text{others}}^2$. \\
    $^{*}$Difference Rewards reduce variance by subtracting a counterfactual baseline $r(s, \boldsymbol{a}) - r(s, \boldsymbol{a}^{-i})$, but the degree of reduction depends on the reward structure (separability). For non-separable rewards, cross-agent interaction terms persist. Computing the counterfactual also requires an exact simulator, which is often intractable at scale. \\
    $^{\S}$Mean-Field methods assume homogeneous agents (identical and exchangeable). Results hold only under this approximation, and convergence is to the mean-field equilibrium, which may differ from the Nash equilibrium of the original $N$-agent game.
\end{sidewaystable}

\subsection*{Key Insights from the Comparison}

\textbf{1. Every existing method sacrifices at least one critical property.}
For large-scale cooperative MARL, we identify five desirable properties: $\mathcal{O}(1)$ variance, $\mathcal{O}(1/\epsilon^3)$ sample complexity, heterogeneous policies, parallel updates, and zero extra learned components. No prior method achieves all five:
\begin{itemize}
    \item \textbf{MAPPO / IPPO}: $\mathcal{O}(N)$ variance $\Rightarrow$ $\mathcal{O}(N/\epsilon^3)$ sample complexity.
    \item \textbf{HAPPO}: sequential updates $\Rightarrow$ $\mathcal{O}(N)$ wall-clock cost per iteration.
    \item \textbf{Mean-Field}: requires agent exchangeability $\Rightarrow$ no heterogeneous policies.
    \item \textbf{QMIX / QPLEX}: extra mixing networks or $\mathcal{O}(N^2)$ attention layers.
    \item \textbf{Diff.\ Rewards}: requires an exact simulator, often intractable at scale.
\end{itemize}
DG-PG satisfies all five simultaneously (Theorems~\ref{thm:variance}--\ref{thm:complexity}), with only a minimal modification to the advantage computation in standard policy gradient algorithm such as MAPPO.

\textbf{2. Leveraging analytical models without sacrificing optimality.}
All other methods in Table~\ref{tab:method_comparison} are purely model-free, either imposing no assumptions at all (MAPPO, IPPO, HAPPO, COMA) or encoding structural constraints into the architecture (QMIX, VDN, Mean-Field). DG-PG is unique in leveraging an \textit{external analytical model} to reduce variance while remaining model-free in its convergence guarantees. 

DG-PG's gradient-level use of analytical models also distinguishes it from other knowledge-informed learning paradigms: Physics-Informed Neural Networks~\citep{raissi2019physics} embed governing equations directly into the loss, requiring exact PDEs; Residual RL~\citep{johannink2019residual} uses a hand-designed controller as a permanent base policy, inheriting its biases; Guided Policy Search~\citep{levine2013guided} relies on trajectory optimization with known dynamics. These methods all incorporate domain knowledge into the objective or policy. In contrast, DG-PG derives the gradient of a guidance term from the analytical model to steer policy updates, leaving the original objective $J$ unchanged. As a result, as long as the reference remains descent-aligned (Assumption~\ref{ass:alignment}), model inaccuracies affect only variance, not the stationary points of the original cooperative objective (Theorem~\ref{thm:invariance}). The adjustable $\alpha$ further allows practitioners to tune the degree of reliance based on model fidelity.

\section{Theoretical Assumptions and Dependencies}
\label{app:assumptions_table}

For reviewer convenience, Table~\ref{tab:assumptions_deps} summarizes the dependency structure between our theoretical results and the assumptions they rely upon.

\begin{table}[h]
    \centering
    \caption{Dependency of theoretical results on assumptions. ``\checkmark'' indicates that the result requires the corresponding assumption. All assumptions are stated in the main text (Section~\ref{subsec:prior_assumptions}) and Appendix~\ref{app:proof_complexity}.}
    \label{tab:assumptions_deps}
    \small
    \begin{tabularx}{\columnwidth}{lccccX}
        \toprule
        \textbf{Result} & \makecell{\textbf{Bounded}\\\textbf{Increments}\\\textbf{(Asm.~\ref{ass:bounded_system_state})}} & \makecell{\textbf{Exogeneity}\\\textbf{(Asm.~\ref{ass:exogeneity})}} & \makecell{\textbf{Alignment}\\\textbf{(Asm.~\ref{ass:alignment})}} & \makecell{\textbf{Smoothness}\\\textbf{\& GD}\\\textbf{(Asm.~\ref{ass:gradient_domination})}} & \textbf{Role of Assumptions} \\
        \midrule
        \makecell[l]{Thm.~\ref{thm:invariance}\\(Stationary Points)} & & \checkmark & \checkmark & & Exogeneity ensures reference does not shift; Alignment ensures guidance gradient vanishes at stationary points. \\
        \midrule
        \makecell[l]{Thm.~\ref{thm:variance}\\(Variance Bound)} & \checkmark & \checkmark & \checkmark & & Bounded increments yield $\mathcal{O}(1)$ guidance variance; Exogeneity makes reference deterministic; Alignment ensures negative correlation $\rho < 0$. \\
        \midrule
        \makecell[l]{Thm.~\ref{thm:complexity}\\(Sample Complexity)} & \checkmark & \checkmark & \checkmark & \checkmark & Smoothness enables the per-iteration bound; gradient domination yields the optimality-gap recursion; $\mathcal{O}(1)$ variance from Thm.~\ref{thm:variance}. \\
        \midrule
        \makecell[l]{Prop.~1 in App.~\ref{app:verify_alignment}\\(Cloud Verification)} & & & & & Standalone verification that Alignment holds for cloud scheduling via Cauchy-Schwarz on the load imbalance functional. \\
        \bottomrule
    \end{tabularx}
\end{table}

\textbf{Discussion of Assumption Strength.}
\begin{itemize}
    \item \textbf{Assumption~\ref{ass:exogeneity} (Exogeneity)} is a design choice rather than a restrictive condition. It is satisfied by construction when the reference is computed from exogenous system parameters (e.g., server capacities, total workload) and treated as a fixed target during gradient computation via the stop-gradient operator $\bot(\cdot)$.
    \item \textbf{Assumption~\ref{ass:alignment} (Descent-Aligned Reference)} is the core structural requirement. It holds when the analytical model captures the dominant performance drivers of the cooperative objective. We verify this rigorously for cloud scheduling in Appendix~\ref{app:verify_alignment} and outline verification templates for traffic, communication, and power grid domains.
    \item \textbf{Assumption~\ref{ass:gradient_domination} (Smoothness \& Gradient Domination)} is standard in the policy-gradient literature \citep{agarwal2021theory} and is required only for the convergence rate result (Theorem~\ref{thm:complexity}). The variance reduction (Theorem~\ref{thm:variance}) and preservation of stationary points (Theorem~\ref{thm:invariance}) hold without these regularity conditions.
\end{itemize}

\section{Case Study: Verification of Structural Assumptions}
\label{app:verification_cloud}

In this appendix, we provide a detailed verification that the structural Assumptions \ref{ass:exogeneity} and \ref{ass:alignment} hold for the cloud resource scheduling task used in our experiments. This serves as a representative case study demonstrating how domain-specific priors can be systematically validated within the DG-PG framework.

\textbf{Applicability to Other Domains.}
While we focus on cloud scheduling, similar verification procedures apply to other control problems where analytical priors are available:
\begin{itemize}
    \item \textit{Communication networks}: Network utility maximization theory provides proportional fairness as a natural reference allocation \citep{kelly1998rate}.
    \item \textit{Power grids}: DC optimal power flow and economic dispatch yield reference generation schedules \citep{wood2013power}.
    \item \textit{Traffic systems}: Wardrop equilibrium and system-optimal routing define reference flow distributions \citep{wardrop1952road}.
\end{itemize}
In each case, the key requirement is that the prior defines a descent direction aligned with the cooperative objective, a property we formally verify below for the cloud scheduling domain using fluid-limit analysis from queueing theory \citep{kurtz1981approximation}.

\subsection{System Model and Objective Approximation}

\textbf{System Setup.}
We study a cooperative multi-agent system for task scheduling in a heterogeneous cluster of \(N\) servers. Tasks are characterized by multi-dimensional resource requirements (e.g., CPU, memory), and servers possess heterogeneous capacities across these dimensions. Let $K$ denote the number of resource types. We model the system state using the full resource utilization vector.

Specifically, let \(\boldsymbol{x}^i_t = (x^{i,1}_t, \dots, x^{i,K}_t)^\top \in \mathbb{R}_{\geq 0}^K\) denote the \textit{resource utilization vector} on server \(i\) at time \(t\), where $x^{i,k}_t$ is the load on the $k$-th resource dimension. The server capacities are given by \(\boldsymbol{\mu}^i = (\mu^{i,1}, \dots, \mu^{i,K})^\top\), where $\mu^{i,k}$ is the maximum processing rate of server $i$ on resource dimension $k$. This generalization captures the coupling between resources, avoiding the information loss associated with scalar bottleneck abstractions.

The total workload pressure in the system on dimension $k$, denoted as \(C^k\), is conserved under task assignment decisions:
\begin{equation*}
    \sum_{i=1}^N x^{i,k}_t = C^k, \quad x^{i,k}_t \geq 0, \; \forall t, \forall k.
\end{equation*}

\textbf{Surrogate Objective via Queueing Theory.}
To construct the reference state $\tilde{\boldsymbol{X}}$, we need an analytically tractable objective whose minimizer serves as the guidance target. The cooperative MARL objective $J$ itself is intractable because it aggregates queueing delay, energy efficiency, and fairness through complex, non-separable interactions. However, queueing theory reveals that improving each of these criteria reduces to a common structural condition:
\begin{itemize}
    \item \textbf{Queueing delay} is a convex function of server utilization, scaling as $1/(\mu - x)$ near capacity \citep{bertsekas2021data}. By Jensen's inequality, total delay is minimized when utilization is equalized across servers.
    \item \textbf{Energy cost} in heterogeneous clusters with varying $\mu^{i,k}$ is minimized by capacity-proportional allocation, which avoids overloading low-efficiency machines.
    \item \textbf{Fairness} improves as utilization variance decreases, reducing service-level violations on congested servers.
\end{itemize}
All three criteria improve when server utilization is balanced proportionally to capacity. We formalize this shared structure as the \textit{load imbalance functional}:
\begin{equation}
    \mathcal{J}_{\text{load}}(\boldsymbol{X}) \triangleq \sum_{k=1}^K \sum_{i=1}^N \frac{1}{\mu^{i,k}} (x^{i,k})^2,
    \label{eq:surrogate_objective}
\end{equation}
where $\boldsymbol{X} = \{\boldsymbol{x}^1, \dots, \boldsymbol{x}^N\}$ is the global system state. The gradient $\frac{\partial \mathcal{J}_{\text{load}}}{\partial x^{i,k}} = \frac{2x^{i,k}}{\mu^{i,k}}$ is proportional to the utilization rate of resource $k$ on server $i$, so minimizing $\mathcal{J}_{\text{load}}$ drives the system toward capacity-proportional allocation. Crucially, $\mathcal{J}_{\text{load}}$ is convex and separable across resource dimensions, admitting a closed-form minimizer. This tractability also enables rigorous verification of Assumption~\ref{ass:alignment}, which would be intractable for the composite objective $J$ directly.

\textbf{Optimal Reference via Load Balancing.}
To construct the reference state $\tilde{\boldsymbol{X}}$, we solve the constrained optimization problem of minimizing $\mathcal{J}_{\text{load}}$ under workload conservation. For each resource dimension $k$, this reduces to:
\begin{equation*}
    \min_{\boldsymbol{X}} \sum_{i=1}^N \frac{1}{\mu^{i,k}} (x^{i,k})^2 \quad \text{subject to} \quad \sum_{i=1}^N x^{i,k} = C^k.
\end{equation*}
This is a classical convex optimization problem whose solution is well-established in queueing theory and resource allocation \citep{bertsekas2021data, ghodsi2011dominant}. Constructing the Lagrangian for dimension $k$:
\begin{equation*}
    \mathcal{L}_{\text{Lag}}^k(\boldsymbol{x}^{\cdot,k}, \lambda_k) = \sum_{i=1}^N \frac{1}{\mu^{i,k}} (x^{i,k})^2 - \lambda_k \left( \sum_{i=1}^N x^{i,k} - C^k \right),
\end{equation*}
and taking the first-order optimality condition:
\begin{equation*}
    \frac{\partial \mathcal{L}_{\text{Lag}}^k}{\partial x^{i,k}} = \frac{2x^{i,k}}{\mu^{i,k}} - \lambda_k = 0 \quad \Longrightarrow \quad x^{i,k} = \frac{\lambda_k \mu^{i,k}}{2}.
\end{equation*}
This shows that at the optimum, load is allocated proportionally to capacity: $x^{i,k} \propto \mu^{i,k}$. Applying the workload constraint $\sum_i x^{i,k} = C^k$ yields:
\begin{equation*}
    x^{*,i,k} = \frac{\mu^{i,k}}{\sum_j \mu^{j,k}} \cdot C^k.
\end{equation*}
Equivalently, the optimal state satisfies the \textit{utilization balancing condition}, i.e., the utilization rate $x^{i,k}/\mu^{i,k}$ is equalized across all servers for each resource type $k$:
\begin{equation*}
    \frac{x^{1,k}}{\mu^{1,k}} = \frac{x^{2,k}}{\mu^{2,k}} = \cdots = \frac{x^{N,k}}{\mu^{N,k}}, \quad \forall k \in \{1, \dots, K\}.
\end{equation*}
This principle is fundamental to heterogeneous resource allocation and aligns with dominant resource fairness (DRF) in multi-dimensional scheduling systems.

\textbf{Theory-Guided Reference State.}
Based on the optimal solution derived above, we define the \textit{capacity-weighted load-balancing equilibrium} as our theory-guided reference:
\begin{equation}
    \tilde{x}^{i,k}_t = \frac{\mu^{i,k}}{\sum_{j=1}^N \mu^{j,k}} \cdot C^k, \quad \forall i, k.
    \label{eq:prior_definition}
\end{equation}
By construction, \(\tilde{\boldsymbol{X}}_t = \{\tilde{\boldsymbol{x}}^1_t, \dots, \tilde{\boldsymbol{x}}^N_t\}\) is the unique global minimizer of \(\mathcal{J}_{\text{load}}(\boldsymbol{X})\) under the workload constraints. Note that \(\tilde{\boldsymbol{X}}_t\) is derived from a continuous (fluid) optimization and may not be exactly reachable under discrete task assignments; however, this does not affect the verification below. We now verify that this reference satisfies Assumptions~\ref{ass:exogeneity} and~\ref{ass:alignment}. Specifically, in Appendix~\ref{app:verify_alignment}, we prove that \((\tilde{\boldsymbol{X}} - \boldsymbol{X})\) is aligned with the descent direction of \(\mathcal{J}_{\text{load}}\), guaranteeing that moving toward the reference reduces load imbalance.

\subsection{Verification of Assumption \ref{ass:exogeneity} (Exogeneity)}

The reference state $\tilde{\boldsymbol{X}}_t$ is computed deterministically from the current aggregate system workload $C^k$ but is treated as a fixed constant during the policy gradient computation step. Specifically, for each server $i$ and resource $k$, we define:
\begin{equation*}
    \tilde{x}^{i,k}_t := \bot \left( \frac{\mu^{i,k}}{\sum_{j=1}^N \mu^{j,k}} \cdot C^k \right).
\end{equation*}
Here, $\bot$ denotes the stop-gradient operator. Since $\tilde{\boldsymbol{X}}_t$ depends only on the current state information (total load) and not on the immediate action $a_t$ selected by the policy $\pi_\theta$, we have $\nabla_{\theta} \tilde{\boldsymbol{X}}_t = \mathbf{0}$, ensuring that no gradient flows back through the reference target. This satisfies Assumption \ref{ass:exogeneity}.

\subsection{Verification of Assumption \ref{ass:alignment} (Directional Alignment)}
\label{app:verify_alignment}

Assumption~\ref{ass:alignment} requires that the direction \((\tilde{\boldsymbol{X}} - \boldsymbol{X})\) is an ascent direction for the state-conditioned cooperative objective $J_{\boldsymbol{\theta}}(\boldsymbol{X})$. For the load-imbalance surrogate to certify this requirement, we explicitly require objective--surrogate directional compatibility:
\begin{equation}
    \left\langle \nabla_{\boldsymbol{X}} \mathcal{J}_{\text{load}}(\boldsymbol{X}),\; \tilde{\boldsymbol{X}} - \boldsymbol{X} \right\rangle < 0
    \quad \Longrightarrow \quad
    \left\langle \nabla_{\boldsymbol{X}} J_{\boldsymbol{\theta}}(\boldsymbol{X}),\; \tilde{\boldsymbol{X}} - \boldsymbol{X} \right\rangle > 0.
    \label{eq:objective_surrogate_compatibility}
\end{equation}
The proposition below establishes the antecedent of this condition.

\begin{proposition}[Alignment with Load Imbalance Reduction]
    For any feasible state $\boldsymbol{X} \neq \tilde{\boldsymbol{X}}$ satisfying the workload constraints, the direction \((\tilde{\boldsymbol{X}} - \boldsymbol{X})\) is aligned with the descent direction of the load imbalance functional:
    \begin{equation*}
        \langle \nabla_{\boldsymbol{X}} \mathcal{J}_{\text{load}}(\boldsymbol{X}), \tilde{\boldsymbol{X}} - \boldsymbol{X} \rangle < 0.
    \end{equation*}
    That is, moving towards $\tilde{\boldsymbol{X}}$ reduces load imbalance.
\end{proposition}

\begin{proof}
    Using the definition of the load imbalance functional in Equation~\ref{eq:surrogate_objective}, we compute the inner product between its gradient and the guidance vector. Since $\mathcal{J}_{\text{load}}$ is separable across resource dimensions, the total inner product decomposes as a sum over resource types $k$:
    \begin{align*}
        \langle \nabla_{\boldsymbol{X}} \mathcal{J}_{\text{load}}, \tilde{\boldsymbol{X}} - \boldsymbol{X} \rangle
         & = \sum_{k=1}^K \sum_{i=1}^N \frac{\partial \mathcal{J}_{\text{load}}}{\partial x^{i,k}} (\tilde{x}^{i,k} - x^{i,k}) \\
         & = \sum_{k=1}^K \underbrace{\sum_{i=1}^N \frac{2x^{i,k}}{\mu^{i,k}} (\tilde{x}^{i,k} - x^{i,k})}_{I_k}.
    \end{align*}

    To establish that the total inner product is negative, it suffices to show that each component $I_k \leq 0$. We analyze a fixed resource dimension $k$. First, observe that by the equilibrium definition in Equation~\ref{eq:prior_definition}, the reference utilization rate is constant across all servers:
    \begin{equation*}
        \frac{\tilde{x}^{i,k}}{\mu^{i,k}} = \frac{C^k}{\sum_{j=1}^N \mu^{j,k}}.
    \end{equation*}
    Substituting this constant ratio into the expression for $I_k$:
    \begin{align*}
        I_k & = 2 \sum_{i=1}^N x^{i,k} \left( \frac{\tilde{x}^{i,k}}{\mu_i^k} - \frac{x^{i,k}}{\mu_i^k} \right) \\
            & = 2 \left( \frac{C^k}{\sum_{j=1}^N \mu_j^k} \right) \sum_{i=1}^N x^{i,k} - 2 \sum_{i=1}^N \frac{(x^{i,k})^2}{\mu^{i,k}} \\
            & = \frac{2 (C^k)^2}{\sum_{j=1}^N \mu^{j,k}} - 2 \sum_{i=1}^N \frac{(x^{i,k})^2}{\mu^{i,k}}.
    \end{align*}
    To determine the sign of $I_k$, we apply the Cauchy-Schwarz inequality to vectors $\boldsymbol{u} = (\frac{x^{i,k}}{\sqrt{\mu^{i,k}}})_i$ and $\boldsymbol{v} = (\sqrt{\mu^{i,k}})_i$:
    \begin{equation*}
        \left( \sum_{i=1}^N u_i v_i \right)^2 \leq \left( \sum_{i=1}^N u_i^2 \right) \left( \sum_{i=1}^N v_i^2 \right).
    \end{equation*}
    Substituting the definitions of $u_i$ and $v_i$:
    \begin{equation*}
        \left( \sum_{i=1}^N x^{i,k} \right)^2 \leq \left( \sum_{i=1}^N \frac{(x^{i,k})^2}{\mu^{i,k}} \right) \left( \sum_{i=1}^N \mu^{i,k} \right).
    \end{equation*}
    Rearranging terms and using $\sum x^{i,k} = C^k$:
    \begin{equation*}
        \sum_{i=1}^N \frac{(x^{i,k})^2}{\mu^{i,k}} \geq \frac{(C^k)^2}{\sum_{j=1}^N \mu^{j,k}}.
    \end{equation*}
    Comparing this with the expression for $I_k$, we see that the second term (actual variance) is always greater than or equal to the first term (minimal variance). Thus, $I_k \leq 0$ for all $k$, with strict inequality if the system is unbalanced. Summing over all dimensions implies:
    \begin{equation*}
        \langle \nabla_{\boldsymbol{X}} \mathcal{J}_{\text{load}}(\boldsymbol{X}), \tilde{\boldsymbol{X}} - \boldsymbol{X} \rangle = \sum_{k=1}^K I_k < 0.
    \end{equation*}
    This proves that the direction toward the reference $\tilde{\boldsymbol{X}}$ strictly decreases the load imbalance functional. Together with the objective--surrogate directional compatibility condition in Equation~\ref{eq:objective_surrogate_compatibility}, this implies $\langle \nabla_{\boldsymbol{X}} J_{\boldsymbol{\theta}}(\boldsymbol{X}), \tilde{\boldsymbol{X}} - \boldsymbol{X} \rangle > 0$, thereby verifying Assumption~\ref{ass:alignment} under this condition.
\end{proof}

\textbf{Generalization to Other Domains.}
The verification procedure demonstrated above serves as a \textit{template} for applying DG-PG to other systems. As outlined in the main text, domains such as communication networks (proportional fairness), power grids (economic dispatch), and traffic systems (Wardrop equilibrium) all admit analytical priors with similar structural properties. In each case, the key steps are: (1) identify a domain-specific performance metric with known convexity properties; (2) derive the reference state as its constrained optimum; (3) verify that the reference defines a descent direction using tools from convex analysis or variational inequalities. This case study establishes the feasibility and rigor of such verification, providing a blueprint for extending DG-PG beyond cloud scheduling.

\textbf{Robustness to Approximate Satisfaction.}
In practice, the structural assumptions may not hold exactly. For instance, the fluid-limit approximation underlying the reference state assumes Poisson arrivals and homogeneous service times, which are often violated in real systems due to heavy-tailed workloads and temporal fluctuations. However, DG-PG remains effective as long as the assumptions hold \textit{approximately}, that is, the reference captures the dominant performance drivers even if the analytical model is idealized. The adjustable guidance weight $\alpha$ provides a natural mechanism for handling model inaccuracies. When the prior is less accurate, reducing $\alpha$ diminishes its influence while preserving the variance reduction benefits. Our experiments (Section~\ref{subsec:alpha}) demonstrate that DG-PG is robust across a range of $\alpha$ values, confirming that directional correctness, rather than exact model fidelity, is sufficient for effective guidance.

\section{Derivation of Local Influence Vector}
\label{app:influence_derivation}

For continuous actions, the local influence vector $\boldsymbol{z}_t^i = \frac{\partial \boldsymbol{x}_{t+1}}{\partial a^i_t}$ is obtained directly by differentiating the system transition with respect to $a_t^i$. This appendix considers the discrete server-selection action in the cloud resource scheduling example.
Since the partial derivative is not directly defined for a discrete action, we obtain $\boldsymbol{z}_t^i$ through a linear relaxation of the server-assignment indicator. This derivation generalizes to other domains with similar additive structure.

We use the notation established in Appendix~\ref{app:verification_cloud}: $N$ servers, $K$ resource dimensions, global utilization vector $\boldsymbol{x}_t \in \mathbb{R}^{NK}$, and per-server utilization $\boldsymbol{x}_t^n \in \mathbb{R}^K$ for $n \in \{1,\dots,N\}$. Each agent $i$ holds a task with resource requirement vector $\boldsymbol{w}^i = (w^{i,1}, \dots, w^{i,K})^\top \in \mathbb{R}^K$, where $w^{i,k}$ is the demand on resource dimension $k$ (e.g., CPU cores, memory). Each agent $i$ takes a discrete action $a_t^i \in \{1,\dots,N\}$, where $a_t^i=n$ means that agent $i$ assigns its task to server $n$. To express this action in the state update of each server, we define $a_{t,n}^i$ as a binary assignment indicator for server $n$:
\begin{equation}
    a_{t,n}^i = \mathbb{I}(a_t^i = n) \qquad \text{and} \qquad
    \sum_{n=1}^N a_{t,n}^i = 1.
\end{equation}
The utilization of server $n$ then evolves as
\begin{equation*}
    \boldsymbol{x}_{t+1}^n = \boldsymbol{x}_t^n + \sum_{l=1}^N \boldsymbol{w}^l a_{t,n}^l.
\end{equation*}
To define the derivative with respect to the discrete action, we relax $a_{t,n}^i$ from $\{0,1\}$ to $[0,1]$. For the realized action $a_t^i = n$, holding the other agents' assignments fixed and differentiating the state update gives
\begin{equation*}
    \boldsymbol{z}_t^i \triangleq \frac{\partial \boldsymbol{x}_{t+1}}{\partial a_{t,n}^i} = (\mathbf{0}, \dots, \underbrace{\boldsymbol{w}^i}_{\text{server } n}, \dots, \mathbf{0})^\top.
\end{equation*}
That is, assigning task $i$ to server $n$ increases server $n$'s load by exactly $\boldsymbol{w}^i$ across all resource types, while leaving all other servers unaffected. 

The linear relaxation for discrete action $a_{t,n}^i$ introduces no approximation error in $\boldsymbol{z}_t^i$. Under the original discrete assignment, changing $a_{t,n}^i$ from $0$ to $1$ produces the global state increment
\begin{equation*}
    \left.\boldsymbol{x}_{t+1}\right|_{a_{t,n}^i=1}
    -
    \left.\boldsymbol{x}_{t+1}\right|_{a_{t,n}^i=0}
    =
    (\mathbf{0}, \dots, \underbrace{\boldsymbol{w}^i}_{\text{server } n}, \dots, \mathbf{0})^\top
    =
    \boldsymbol{z}_t^i.
\end{equation*}
The equality above shows that the local influence vector obtained from the relaxed derivative exactly matches the global-state increment induced by the corresponding discrete assignment, confirming that the relaxation introduces no approximation error in \(\boldsymbol{z}_t^i\).

\textbf{Projected coefficient.} In the DG-PG update (Equations~\ref{eq:guidance_coefficient}--\ref{eq:dg_estimator}), each agent's guidance signal is scaled by the inner product $\langle \boldsymbol{x}_{t+1} - \tilde{\boldsymbol{x}}_{t+1}, \boldsymbol{z}_t^i \rangle$, which measures how much agent $i$'s action contributes to the change in deviation from the reference. Substituting the influence vector $\boldsymbol{z}_t^i$ yields:
\begin{equation*}
    \langle \boldsymbol{x}_{t+1} - \tilde{\boldsymbol{x}}_{t+1}, \boldsymbol{z}_t^i \rangle = \sum_{k=1}^K w^{i,k} \cdot (x_{t+1}^{n,k} - \tilde{x}_{t+1}^{n,k}).
\end{equation*}
Each term measures how overloaded server $n$ is on resource $k$, weighted by the task's own demand $w^{i,k}$, so resource types that the task consumes most heavily dominate the signal. Importantly, this coefficient provides a \textit{direction} for policy improvement, not a hard constraint that forces convergence to $\tilde{\boldsymbol{X}}$. The actual system trajectory is determined by the learned policy interacting with the environment.

\subsection{General Applicability}

The key requirement is that the influence vector $\boldsymbol{z}_t^i$ captures the marginal effect of an action on the system state, regardless of whether the action space is continuous or discrete. For discrete actions, $\boldsymbol{z}_t^i$ directly represents the action-induced state change, which can be obtained exactly from the system transition. The system state $\boldsymbol{x}_t$ and influence vector $\boldsymbol{z}_t^i$ can be similarly defined in many other multi-agent systems:
\begin{itemize}
    \item \textbf{Traffic routing}: $\boldsymbol{x}_t$ represents link flows and actions are path selections, so $\boldsymbol{z}_t^i$ indicates which links are used.
    \item \textbf{Inventory management}: $\boldsymbol{x}_t$ represents stock levels and actions are replenishment orders, where $\boldsymbol{z}_t^i$ adds to specific stock dimensions.
    \item \textbf{Communication networks}: $\boldsymbol{x}_t$ represents channel interference levels and actions are power allocations, and $\boldsymbol{z}_t^i$ captures each agent's marginal contribution to interference.
\end{itemize}
In each case, $\boldsymbol{z}_t^i$ enables the projection-based variance reduction by isolating the relevant components of the global state deviation.

\section{Detailed Proof of Preservation of Stationary Points (Theorem~\ref{thm:invariance})}
\label{app:proof_invariance}

In this appendix, we provide the detailed proof for Theorem~\ref{thm:invariance}, establishing that our DG-PG framework preserves the stationary points of the original cooperative objective.

\textbf{Theorem~\ref{thm:invariance} (Preservation of Stationary Points).} \textit{Let $\boldsymbol{\theta}^*$ be a joint policy parameter. If Assumptions~\ref{ass:exogeneity}--\ref{ass:alignment} hold and the policies are sufficiently expressive, then $\boldsymbol{\theta}^*$ is a stationary point of the original cooperative objective if and only if it is a stationary point of the DG-PG objective, i.e., $\nabla_{\theta^i} J(\boldsymbol{\theta}^*) = \mathbf{0}$ for all $i \in \mathcal{N}$ if and only if $\nabla_{\theta^i} J_{\alpha}(\boldsymbol{\theta}^*) = \mathbf{0}$ for all $i \in \mathcal{N}$ and every $\alpha \in (0, 1)$.}

\begin{remark}[Sufficient Expressiveness]
    In Theorem~\ref{thm:invariance}, sufficient policy expressiveness means that the Jacobian $\nabla_{\theta^i} \boldsymbol{x}$ has full row rank for at least one agent $i$, ensuring that a non-zero state-space gradient $\nabla_{\boldsymbol{x}} J_{\boldsymbol{\theta}}(\boldsymbol{x})$ lifts to a non-zero parameter-space gradient $\nabla_{\theta^i} J$. This is a standard condition for overparameterized neural network policies ($p \gg m$) \citep{zhang2019fast,oymak2019overparameterized,balaji2021understanding,soori2021tengrad} and holds generically in the absence of degenerate symmetries.
    \label{remark:sufficient_expressiveness}
\end{remark}

To prove this theorem, we first establish a key lemma stating that the gradient of the guidance term vanishes at any stationary point of the original objective.

\subsection{Lemma: Guidance Consistency}

\begin{lemma}
    \label{lemma:guidance_vanishing}
    Under Assumptions~\ref{ass:exogeneity} and~\ref{ass:alignment}, if $\nabla_{\theta^i} J(\boldsymbol{\theta}^*) = \mathbf{0}$ for all $i \in \mathcal{N}$, then $\nabla_{\theta^i} \mathcal{G}(\boldsymbol{\theta}^*) = \mathbf{0}$ for all $i \in \mathcal{N}$.
\end{lemma}

\begin{proof}
    We proceed by contradiction. Suppose $\nabla_{\theta^i} J(\boldsymbol{\theta}^*) = \mathbf{0}$ for all $i \in \mathcal{N}$ (the stationarity condition), but $\nabla_{\theta^i} \mathcal{G}(\boldsymbol{\theta}^*) \neq \mathbf{0}$ for some agent $i$. Recall that $\mathcal{G}(\boldsymbol{\theta}) = \mathbb{E}_{s \sim \nu^{\boldsymbol{\pi}}}[d(\boldsymbol{x}^{\boldsymbol{\pi}}|_s, \tilde{\boldsymbol{x}}|_s)]$ is the expected deviation under the state visitation distribution $\nu^{\boldsymbol{\pi}}$, where $\boldsymbol{x}^{\boldsymbol{\pi}}|_s$ is the system state jointly determined by agents' actions under policy $\boldsymbol{\pi}$ at MDP state $s$, and $\tilde{\boldsymbol{x}}|_s$ is the corresponding reference. We write $\boldsymbol{x}^* \triangleq \boldsymbol{x}^{\boldsymbol{\pi}^*}|_s$ and $\tilde{\boldsymbol{x}} \triangleq \tilde{\boldsymbol{x}}|_s$ for a given MDP state $s$, where $\boldsymbol{\pi}^* \triangleq \boldsymbol{\pi}_{\boldsymbol{\theta}^*}$.

    \emph{Step 1.} If $\boldsymbol{x}^* = \tilde{\boldsymbol{x}}$ for all states $s$ visited under $\nu^{\boldsymbol{\pi}^*}$, then $d(\boldsymbol{x}^*, \tilde{\boldsymbol{x}}) = 0$ everywhere and $\mathcal{G}(\boldsymbol{\theta}^*) = 0$. This is the global minimum of $\mathcal{G} \geq 0$, so the first-order condition gives $\nabla_{\theta^i} \mathcal{G}(\boldsymbol{\theta}^*) = \mathbf{0}$ for all $i$, contradicting our assumption. Therefore there exists a visited state $s$ at which $\boldsymbol{x}^* \neq \tilde{\boldsymbol{x}}$.

    \emph{Step 2.} By Assumption~\ref{ass:alignment}, at any state $s$ where $\boldsymbol{x}^* \neq \tilde{\boldsymbol{x}}$:
    \begin{equation*}
        \langle \nabla_{\boldsymbol{x}} J_{\boldsymbol{\theta}^*}(\boldsymbol{x}^*),\, \tilde{\boldsymbol{x}} - \boldsymbol{x}^* \rangle > 0.
    \end{equation*}
    Since $\tilde{\boldsymbol{x}} - \boldsymbol{x}^* \neq \mathbf{0}$, the state-space gradient $\nabla_{\boldsymbol{x}} J_{\boldsymbol{\theta}^*}(\boldsymbol{x}^*)$ must be non-zero.

    \emph{Step 3.} Under sufficient policy expressiveness (see Remark~\ref{remark:sufficient_expressiveness}), a non-zero state-space gradient $\nabla_{\boldsymbol{x}} J_{\boldsymbol{\theta}^*}(\boldsymbol{x}^*) \neq \mathbf{0}$ implies $\nabla_{\theta^i} J(\boldsymbol{\theta}^*) \neq \mathbf{0}$ for some agent $i$. This contradicts the stationarity condition.

    Therefore $\nabla_{\theta^i} \mathcal{G}(\boldsymbol{\theta}^*) = \mathbf{0}$ for all $i \in \mathcal{N}$.
\end{proof}

\subsection{Proof of Theorem~\ref{thm:invariance}}

\begin{proof}
    At a stationary point of the original cooperative objective $\boldsymbol{\theta}^*$, we have $\nabla_{\theta^i} J(\boldsymbol{\theta}^*) = \mathbf{0}$ for all $i \in \mathcal{N}$. By Lemma~\ref{lemma:guidance_vanishing}, $\nabla_{\theta^i} \mathcal{G}(\boldsymbol{\theta}^*) = \mathbf{0}$ for all $i \in \mathcal{N}$. Therefore:
    \begin{equation*}
        \nabla_{\theta^i} J_{\alpha}(\boldsymbol{\theta}^*) = (1-\alpha) \nabla_{\theta^i} J(\boldsymbol{\theta}^*) - \alpha \nabla_{\theta^i} \mathcal{G}(\boldsymbol{\theta}^*) = \mathbf{0},
    \end{equation*}
    confirming that the DG-PG gradient also vanishes at every stationary point of the original cooperative objective.

    For the converse, suppose that $\nabla_{\theta^i} J_{\alpha}(\boldsymbol{\theta}^*) = \mathbf{0}$ for all $i \in \mathcal{N}$. By the definition of the DG-PG gradient, this condition gives
    \begin{align*}
        \mathbf{0}
        &=
        \nabla_{\theta^i} J_{\alpha}(\boldsymbol{\theta}^*) \\
        &=
        (1-\alpha)\nabla_{\theta^i} J(\boldsymbol{\theta}^*)
        -
        \alpha\nabla_{\theta^i} \mathcal{G}(\boldsymbol{\theta}^*),
        \quad \forall i \in \mathcal{N}.
    \end{align*}
    Rearranging the final equality yields
    \begin{equation*}
        (1-\alpha)\nabla_{\theta^i} J(\boldsymbol{\theta}^*)
        =
        \alpha\nabla_{\theta^i} \mathcal{G}(\boldsymbol{\theta}^*),
        \quad \forall i \in \mathcal{N}.
    \end{equation*}
    Taking the inner product of both sides with $\nabla_{\theta^i} J(\boldsymbol{\theta}^*)$ gives
    \begin{equation*}
        (1-\alpha)
        \|\nabla_{\theta^i} J(\boldsymbol{\theta}^*)\|^2
        =
        \alpha
        \left\langle
        \nabla_{\theta^i} J(\boldsymbol{\theta}^*),
        \nabla_{\theta^i} \mathcal{G}(\boldsymbol{\theta}^*)
        \right\rangle,
        \quad \forall i \in \mathcal{N}.
    \end{equation*}
    Summing these equalities over all agents yields
    \begin{align}
        (1-\alpha)
        \sum_{i \in \mathcal{N}}
        \|\nabla_{\theta^i} J(\boldsymbol{\theta}^*)\|^2
        &=
        \alpha
        \sum_{i \in \mathcal{N}}
        \left\langle
        \nabla_{\theta^i} J(\boldsymbol{\theta}^*),
        \nabla_{\theta^i} \mathcal{G}(\boldsymbol{\theta}^*)
        \right\rangle \notag \\
        &=
        \alpha
        \left\langle
        \nabla_{\boldsymbol{\theta}} J(\boldsymbol{\theta}^*),
        \nabla_{\boldsymbol{\theta}} \mathcal{G}(\boldsymbol{\theta}^*)
        \right\rangle.
        \label{eq:converse_summed_identity}
    \end{align}
    Under Assumption~\ref{ass:alignment}, a local policy improvement along $\nabla_{\boldsymbol{\theta}} J(\boldsymbol{\theta}^*)$ induces a system-state change toward the reference and therefore does not increase the expected deviation $\mathcal{G}$ to first order. Equivalently, the directional derivative of $\mathcal{G}$ along $\nabla_{\boldsymbol{\theta}} J(\boldsymbol{\theta}^*)$ is non-positive, yielding $\langle \nabla_{\boldsymbol{\theta}} J(\boldsymbol{\theta}^*), \nabla_{\boldsymbol{\theta}} \mathcal{G}(\boldsymbol{\theta}^*) \rangle \le 0$.
    Combining the gradient-alignment inequality with Equation~\ref{eq:converse_summed_identity} and using $\alpha \in (0,1)$ gives
    \begin{equation*}
        0
        \le
        (1-\alpha)
        \sum_{i \in \mathcal{N}}
        \|\nabla_{\theta^i} J(\boldsymbol{\theta}^*)\|^2
        =
        \alpha
        \left\langle
        \nabla_{\boldsymbol{\theta}} J(\boldsymbol{\theta}^*),
        \nabla_{\boldsymbol{\theta}} \mathcal{G}(\boldsymbol{\theta}^*)
        \right\rangle
        \le
        0.
    \end{equation*}
    Therefore,
    \begin{equation*}
        \sum_{i \in \mathcal{N}}
        \|\nabla_{\theta^i} J(\boldsymbol{\theta}^*)\|^2
        =
        0.
    \end{equation*}
    Since every term in this sum is non-negative, $\nabla_{\theta^i} J(\boldsymbol{\theta}^*) = \mathbf{0}$ for every $i \in \mathcal{N}$. 
    
    Combining the forward and converse results yields
    \begin{equation*}
        \nabla_{\theta^i} J(\boldsymbol{\theta}^*) = \mathbf{0}, \quad \forall i \in \mathcal{N}
        \quad \Longleftrightarrow \quad
        \nabla_{\theta^i} J_{\alpha}(\boldsymbol{\theta}^*) = \mathbf{0}, \quad \forall i \in \mathcal{N}.
    \end{equation*}
    Thus, the guidance term introduces no additional stationary points.
\end{proof}

\section{Variance Decomposition and Reduction Analysis (Theorem~\ref{thm:variance})}
\label{app:proof_variance_complexity}

In this appendix, we prove that the DG-PG estimator achieves agent-independent variance scaling. As established in Section~\ref{sec:variance_analysis}, the standard policy gradient estimator has variance $\sigma_J^2 = \mathcal{O}(N)$ due to cross-agent exploration noise \citep{kuba2021settling}, while the gradient estimator for the guidance term has variance $\sigma_{\mathcal{G}}^2 = \mathcal{O}(1)$ under Assumption~\ref{ass:bounded_system_state} (Lemma~\ref{lemma:guidance_variance}).

\textbf{Theorem~\ref{thm:variance} (Agent-Independent Variance).} \textit{Under Assumptions~\ref{ass:bounded_system_state}--\ref{ass:alignment}, let $\sigma_J^2 \triangleq \mathrm{Var}(\hat{g}_J^i)$, $\sigma_{\mathcal{G}}^2 \triangleq \mathrm{Var}(\hat{g}_{\mathcal{G}}^i)$, and $\rho \triangleq \mathrm{Corr}(\hat{g}_J^i, \hat{g}_{\mathcal{G}}^i) \in (-1, 0)$. The minimum variance over the guidance weight $\alpha$ is}
\begin{equation*}
    \min_{\alpha \in (0,1)} \mathrm{Var}(\hat{g}_{\mathrm{DG}}^i) = \frac{\sigma_J^2 \, \sigma_{\mathcal{G}}^2 \,(1 - \rho^2)}{\sigma_J^2 + \sigma_{\mathcal{G}}^2 + 2\rho\,\sigma_J \sigma_{\mathcal{G}}}.
\end{equation*}
\textit{Since $\sigma_J^2 = \mathcal{O}(N)$ and $\sigma_{\mathcal{G}}^2 = \mathcal{O}(1)$, this minimum variance is $\mathcal{O}(1)$, independent of $N$.}

\subsection{Bound on the Guidance Variance}

{
\begin{lemma}
    \label{lemma:guidance_variance}
    Under Assumption~\ref{ass:bounded_system_state}, if the score function $\nabla_{\theta^i}\log\pi^i$ is bounded, then $\sigma_{\mathcal{G}}^2=\mathcal{O}(1)$.
\end{lemma}

\begin{proof}
    The guidance coefficient in Equation~\ref{eq:guidance_coefficient} decomposes as
    \begin{equation*}
        \langle \boldsymbol{x}_{t+1} - \tilde{\boldsymbol{x}}_{t+1}, \boldsymbol{z}_t^i \rangle
        =
        \langle \boldsymbol{x}_t - \tilde{\boldsymbol{x}}_{t+1}, \boldsymbol{z}_t^i \rangle
        +
        \langle \boldsymbol{x}_{t+1} - \boldsymbol{x}_t, \boldsymbol{z}_t^i \rangle.
    \end{equation*}
    We bound each part of the guidance coefficient in turn.
    \begin{itemize}
        \item \textit{Influence vector.} $\boldsymbol{z}_t^i$ is determined by the system transition, and its nonzero entries, which correspond to the dimensions affected by agent $i$, are bounded in number and magnitude independently of $N$.
        \item \textit{First term.} $\boldsymbol{x}_t$ is fixed given $s_t$, and the reference $\tilde{\boldsymbol{x}}_{t+1}$ is computed independently of the policies. This term therefore involves no other agents' actions and does not grow with $N$.
        \item \textit{Second term.} This term involves only the one-step increments $x_{t+1}^k - x_t^k$ on the dimensions affected by agent $i$, whose second moments are bounded by $C_x$ under Assumption~\ref{ass:bounded_system_state}.
    \end{itemize}
    The guidance coefficient therefore has an $\mathcal{O}(1)$ second moment. The bounded score function is a standard assumption in policy-gradient analysis and holds for common policies, such as softmax policies with bounded features or Gaussian policies with bounded actions and a lower-bounded variance. Multiplying the coefficient by the score function gives $\sigma_{\mathcal{G}}^2 \le \mathbb{E}[\|\hat{g}_{\mathcal{G},t}^i\|^2] = \mathcal{O}(1)$.
\end{proof}
}

\subsection{Proof of Theorem~\ref{thm:variance}}

\begin{proof}
    For the DG-PG estimator $\hat{g}_{\text{DG}}^i = (1-\alpha) \hat{g}_{J}^i - \alpha \hat{g}_{\mathcal{G}}^i$, the variance is:
    \begin{align}
        \sigma_{\text{DG}}^2(\alpha) = (1-\alpha)^2 \sigma_J^2 + \alpha^2 \sigma_{\mathcal{G}}^2 - 2\alpha(1-\alpha) \rho \sigma_J \sigma_{\mathcal{G}},
        \label{eq:variance_combined}
    \end{align}
    where $\rho = \text{Corr}(\hat{g}_J^i, \hat{g}_{\mathcal{G}}^i) \in (-1, 0)$ by Assumption~\ref{ass:alignment}.

    The estimator variance $\sigma_{\text{DG}}^2(\alpha)$ is a convex quadratic in $\alpha$. To minimize it, we set $\frac{d}{d\alpha}\sigma_{\text{DG}}^2 = 0$:
    \begin{align}
        -2(1-\alpha) \sigma_J^2 + 2\alpha \sigma_{\mathcal{G}}^2 - 2(1-2\alpha) \rho \sigma_J \sigma_{\mathcal{G}} = 0
        \quad \Longrightarrow \quad
        \alpha^* = \frac{\sigma_J^2 + \rho \, \sigma_J \sigma_{\mathcal{G}}}{\sigma_J^2 + \sigma_{\mathcal{G}}^2 + 2\rho \, \sigma_J \sigma_{\mathcal{G}}}.
        \label{eq:optimal_alpha}
    \end{align}

    Substituting $\alpha^*$ back into $\sigma_{\text{DG}}^2(\alpha)$, the minimum variance is:
    \begin{align}
        \sigma_{\text{DG}, \min}^2 = \frac{\sigma_J^2 \sigma_{\mathcal{G}}^2 (1-\rho^2)}{\sigma_J^2 + \sigma_{\mathcal{G}}^2 + 2\rho \sigma_J \sigma_{\mathcal{G}}}.
        \label{eq:min_variance}
    \end{align}

    As established in Section~\ref{sec:variance_analysis}, $\sigma_J^2 = \mathcal{O}(N)$ because the return aggregates all $N$ agents' stochastic actions, while $\sigma_{\mathcal{G}}^2 = \mathcal{O}(1)$ by Lemma~\ref{lemma:guidance_variance}. Dividing numerator and denominator by $\sigma_J^2$:
    \begin{align*}
        \sigma_{\text{DG}, \min}^2
        = \frac{\sigma_{\mathcal{G}}^2 (1-\rho^2)}{1 + \frac{\sigma_{\mathcal{G}}^2}{\sigma_J^2} + 2\rho \frac{\sigma_{\mathcal{G}}}{\sigma_J}}
        = \frac{\sigma_{\mathcal{G}}^2 (1-\rho^2)}{1 + \mathcal{O}(N^{-1})}.
    \end{align*}
    The numerator $\sigma_{\mathcal{G}}^2 (1-\rho^2) = \mathcal{O}(1)$ since $\sigma_{\mathcal{G}}^2 = \mathcal{O}(1)$ and $\rho \in (-1, 0)$ is a constant. The denominator $1 + \mathcal{O}(N^{-1})$ is bounded away from zero for all $N \geq 1$. Therefore:
    \begin{align}
        \sigma_{\text{DG}, \min}^2 = \mathcal{O}(1), \quad \text{for all } N \geq 1.
        \label{eq:variance_limit}
    \end{align}
\end{proof}

\begin{remark}
    The optimal mixing coefficient $\alpha^* \to 1$ as $N \to \infty$, indicating that large-scale systems should rely primarily on the guidance signal to suppress cross-agent variance. In practice, we adopt a dynamic $\alpha$ schedule for cloud scheduling with $N \leq 50$, starting with a large $\alpha$ for fast early convergence and decaying to a smaller value to allow the RL signal to fine-tune beyond the reference, and fix $\alpha{=}0.9$ for larger scales and for power dispatch (see Section~\ref{sec:experiments}). 
\end{remark}

\section{Convergence Rate and Sample Complexity (Theorem~\ref{thm:complexity})}
\label{app:proof_complexity}

We establish the sample complexity of DG-PG under standard regularity conditions.

\textbf{Theorem~\ref{thm:complexity} (Sample Complexity).} \textit{Under the smoothness and gradient domination conditions (Assumption~\ref{ass:gradient_domination} below), for a fixed guidance weight $\alpha \in (0,1)$, and with $\sigma_{\text{DG}}^2 = \mathcal{O}(1)$ from Theorem~\ref{thm:variance}, DG-PG achieves an $\epsilon$-optimal policy in $T = \mathcal{O}\bigl(\frac{L \kappa_{\mathrm{GD}}^4}{\epsilon^3}\bigr)$ iterations, independent of $N$.}


\begin{assumption}[Smoothness and Gradient Domination]
    \label{ass:gradient_domination}
    The cooperative objective $J(\boldsymbol{\theta})$ is $L$-smooth and satisfies the gradient domination condition with constant $\kappa_{\mathrm{GD}} > 0$:
    \begin{equation*}
        J^* - J(\boldsymbol{\theta})
        \le
        \kappa_{\mathrm{GD}}\|\nabla J(\boldsymbol{\theta})\|.
    \end{equation*}
    where $J^* \triangleq \sup_{\boldsymbol{\theta}} J(\boldsymbol{\theta})$ denotes the optimal value of the cooperative objective.
\end{assumption}
Together, these conditions guarantee global convergence of policy-gradient methods to an optimal policy. Gradient domination is a standard condition in policy-gradient analysis that connects first-order stationarity with global optimality \citep{agarwal2021theory}. It holds for tabular policies under sufficient state and action coverage, characterized by a finite distribution-mismatch coefficient for direct parameterizations and non-vanishing optimal-action probabilities for softmax parameterizations \citep{agarwal2021theory,mei2020global}.

\subsection{Proof of Theorem~\ref{thm:complexity}}

\begin{proof}
    Our proof first characterizes, through $L$-smoothness and the update rule, how the variance and bias of $\hat{g}_{\text{DG}}$ affect the per-iteration change in $J$. We then apply the gradient domination condition to connect the gradient norm $\|\nabla J\|^2$ to the optimality gap $J^* - J(\boldsymbol{\theta}_k)$, thereby obtaining the $N$-independent iteration complexity. The following analysis holds for a \textit{fixed} mixing parameter $\alpha \in (0,1)$. When $\alpha$ is varied dynamically in practice, the convergence guarantee applies to each phase as a snapshot analysis. 

    \textit{Step 1: Smoothness bound.} Let $\bar{g}_{\text{DG}}(\boldsymbol{\theta}_k) \triangleq \mathbb{E}_k[\hat{g}_{\text{DG}}(\boldsymbol{\theta}_k)]$ denote the expected augmented gradient. Consider the stochastic gradient ascent update
    \begin{align}
        \boldsymbol{\theta}_{k+1} = \boldsymbol{\theta}_k + \eta \hat{g}_{\text{DG}}(\boldsymbol{\theta}_k)
        \label{eq:update_rule}
    \end{align}
    By $L$-smoothness of $J$:
    \begin{align*}
        J(\boldsymbol{\theta}_{k+1})
        \ge J(\boldsymbol{\theta}_k)
        + \langle \nabla J(\boldsymbol{\theta}_k),\, \boldsymbol{\theta}_{k+1} - \boldsymbol{\theta}_k \rangle
        - \frac{L}{2}\|\boldsymbol{\theta}_{k+1} - \boldsymbol{\theta}_k\|^2.
    \end{align*}
    Substituting $\boldsymbol{\theta}_{k+1} - \boldsymbol{\theta}_k = \eta \hat{g}_{\text{DG}}(\boldsymbol{\theta}_k)$ and taking the conditional expectation $\mathbb{E}_k[\cdot] \triangleq \mathbb{E}[\cdot \,|\, \boldsymbol{\theta}_k]$:
    \begin{align}
        \mathbb{E}_k[J(\boldsymbol{\theta}_{k+1})]
        \ge J(\boldsymbol{\theta}_k)
        + \eta \langle \nabla J(\boldsymbol{\theta}_k),\, \bar{g}_{\text{DG}}(\boldsymbol{\theta}_k) \rangle
        - \frac{L\eta^2}{2} \mathbb{E}_k[\|\hat{g}_{\text{DG}}(\boldsymbol{\theta}_k)\|^2].
        \label{eq:descent_lemma}
    \end{align}

    \textit{Step 2: Per-iteration descent bound.} The descent inequality Equation~\ref{eq:descent_lemma} shows that the per-iteration ascent on $J$ is governed by two quantities: the inner product term $\langle \nabla J, \bar{g}_{\text{DG}} \rangle$, which captures the alignment between the true gradient and the augmented gradient, and the second moment term $\mathbb{E}_k[\|\hat{g}_{\text{DG}}\|^2]$, which captures the estimator's variance and bias. We bound both independently of $N$.

    For the inner product term $\langle \nabla J(\boldsymbol{\theta}_k),\, \bar{g}_{\text{DG}}(\boldsymbol{\theta}_k) \rangle$, by definition of the expected augmented gradient:
    \begin{align*}
        \bar{g}_{\text{DG}}(\boldsymbol{\theta}_k) = (1-\alpha)\nabla J(\boldsymbol{\theta}_k) - \alpha \nabla \mathcal{G}(\boldsymbol{\theta}_k).
    \end{align*}
    Therefore:
    \begin{align*}
        \langle \nabla J(\boldsymbol{\theta}_k),\, \bar{g}_{\text{DG}}(\boldsymbol{\theta}_k) \rangle
        = (1-\alpha)\|\nabla J(\boldsymbol{\theta}_k)\|^2 - \alpha \langle \nabla J(\boldsymbol{\theta}_k),\, \nabla \mathcal{G}(\boldsymbol{\theta}_k) \rangle.
    \end{align*}
    Since improving $J$ reduces the deviation $\mathcal{G}$ (Assumption~\ref{ass:alignment}), the true gradients satisfy $\langle \nabla J(\boldsymbol{\theta}_k), \nabla \mathcal{G}(\boldsymbol{\theta}_k) \rangle \le 0$, so $-\alpha \langle \nabla J, \nabla \mathcal{G} \rangle \ge 0$. Dropping this non-negative term yields the lower bound for the inner product term in Equation \ref{eq:descent_lemma}:
    \begin{align}
        \langle \nabla J(\boldsymbol{\theta}_k),\, \bar{g}_{\text{DG}}(\boldsymbol{\theta}_k) \rangle \ge (1-\alpha) \|\nabla J(\boldsymbol{\theta}_k)\|^2.
        \label{eq:inner_product_bound}
    \end{align}

    For the second moment $\mathbb{E}_k[\|\hat{g}_{\text{DG}}(\boldsymbol{\theta}_k)\|^2]$, by the bias-variance decomposition:
    \begin{align}
        \mathbb{E}_k[\|\hat{g}_{\text{DG}}(\boldsymbol{\theta}_k)\|^2] = \|\bar{g}_{\text{DG}}(\boldsymbol{\theta}_k)\|^2 + \sigma_{\text{DG}}^2.
        \label{eq:bias_variance}
    \end{align}
    The variance $\sigma_{\text{DG}}^2 = \mathcal{O}(1)$ by Theorem~\ref{thm:variance}. Since $\bar{g}_{\text{DG}} = (1-\alpha)\nabla J - \alpha \nabla \mathcal{G}$, applying the triangle inequality:
    \begin{align}
        \|\bar{g}_{\text{DG}}(\boldsymbol{\theta}_k)\|^2
        \le 2(1-\alpha)^2 \|\nabla J(\boldsymbol{\theta}_k)\|^2 + 2\alpha^2 \|\nabla \mathcal{G}(\boldsymbol{\theta}_k)\|^2.
        \label{eq:triangle_bound}
    \end{align}
    For the $\|\nabla \mathcal{G}(\boldsymbol{\theta}_k)\|$ term, by Equation~\ref{eq:gradient_decomposition} the gradient of the guidance term has the form $\nabla_{\theta^i} \mathcal{G} = \mathbb{E}[\langle \boldsymbol{x}_{t+1} - \tilde{\boldsymbol{x}}_{t+1}, \boldsymbol{z}_t^i \rangle \nabla_{\theta^i} \log \pi^i]$. Each factor is bounded:
    \begin{itemize}
        \item $\|\boldsymbol{x}_{t+1} - \tilde{\boldsymbol{x}}_{t+1}\|$ is bounded by compactness of the system state space.
        \item $\|\boldsymbol{z}_t^i\|$ is bounded because each agent has bounded marginal effect on the system state.
        \item $\|\nabla_{\theta^i} \log \pi^i\|$ is bounded by standard score function regularity.
    \end{itemize}
    Applying Cauchy--Schwarz to the inner product $\langle \boldsymbol{x}_{t+1} - \tilde{\boldsymbol{x}}_{t+1}, \boldsymbol{z}_t^i \rangle$ and combining with the score function bound, there exists a constant $B > 0$, independent of $N$, such that:
    \begin{align}
        \|\nabla \mathcal{G}(\boldsymbol{\theta}_k)\| \le B.
        \label{eq:gradient_bound}
    \end{align}
    Combining Equations~\ref{eq:gradient_bound} and~\ref{eq:triangle_bound} with Equation~\ref{eq:bias_variance}, we obtain:
    \begin{align}
        \mathbb{E}_k[\|\hat{g}_{\text{DG}}(\boldsymbol{\theta}_k)\|^2]
        \le 2(1-\alpha)^2 \|\nabla J(\boldsymbol{\theta}_k)\|^2 + 2\alpha^2 B^2 + \sigma_{\text{DG}}^2.
        \label{eq:second_moment_bound}
    \end{align}
    which is the upper bound for the second moment term in Equation \ref{eq:descent_lemma}.

    Substituting the inner product bound (Equation~\ref{eq:inner_product_bound}) and the second moment bound (Equation~\ref{eq:second_moment_bound}) into the descent inequality (Equation~\ref{eq:descent_lemma}), we obtain the per-iteration descent bound:
    \begin{align*}
        \mathbb{E}_k[J(\boldsymbol{\theta}_{k+1})]
        &\ge J(\boldsymbol{\theta}_k)
        + \eta(1-\alpha) \|\nabla J(\boldsymbol{\theta}_k)\|^2
        - \frac{L\eta^2}{2} \left( 2(1-\alpha)^2 \|\nabla J(\boldsymbol{\theta}_k)\|^2 + 2\alpha^2 B^2 + \sigma_{\text{DG}}^2 \right) \\
        &= J(\boldsymbol{\theta}_k)
        + \left[\eta(1-\alpha) - L\eta^2(1-\alpha)^2\right] \|\nabla J(\boldsymbol{\theta}_k)\|^2
        - \frac{L\eta^2}{2}(2\alpha^2 B^2 + \sigma_{\text{DG}}^2).
    \end{align*}
    Under the update rule (Equation~\ref{eq:update_rule}), for the gradient signal to contribute positively to the expected change in $J$, the coefficient of $\|\nabla J(\boldsymbol{\theta}_k)\|^2$ must be positive. Choosing $\eta \le \frac{1}{2L(1-\alpha)}$ ensures $\eta(1-\alpha) - L\eta^2(1-\alpha)^2 \ge \frac{\eta(1-\alpha)}{2}$:
    \begin{align}
        \mathbb{E}_k[J(\boldsymbol{\theta}_{k+1})]
        \ge J(\boldsymbol{\theta}_k)
        + \frac{\eta(1-\alpha)}{2} \|\nabla J(\boldsymbol{\theta}_k)\|^2
        - \frac{L\eta^2}{2}(2\alpha^2 B^2 + \sigma_{\text{DG}}^2).
        \label{eq:per_iteration_descent}
    \end{align}

    \textit{Step 3: Gradient domination and recursion.} We now use the gradient domination condition to convert the per-iteration bound in Equation~\ref{eq:per_iteration_descent} into a recursion for the optimality gap.

    Subtracting both sides of Equation~\ref{eq:per_iteration_descent} from $J^*$ gives:
    \begin{align*}
        J^* - \mathbb{E}_k[J(\boldsymbol{\theta}_{k+1})]
        \le J^* - J(\boldsymbol{\theta}_k)
        - \frac{\eta(1-\alpha)}{2} \|\nabla J(\boldsymbol{\theta}_k)\|^2
        + \frac{L\eta^2}{2}(2\alpha^2 B^2 + \sigma_{\text{DG}}^2).
    \end{align*}
    Taking full expectation and defining $\Delta_k \triangleq \mathbb{E}[J^* - J(\boldsymbol{\theta}_k)]$:
    \begin{align*}
        \Delta_{k+1}
        \le \Delta_k
        - \frac{\eta(1-\alpha)}{2}
        \mathbb{E}[\|\nabla J(\boldsymbol{\theta}_k)\|^2]
        + \frac{L\eta^2}{2}(2\alpha^2 B^2 + \sigma_{\text{DG}}^2).
    \end{align*}
    Since $J^* - J(\boldsymbol{\theta}_k) \ge 0$, the gradient domination condition (Assumption~\ref{ass:gradient_domination}) implies:
    \begin{align*}
        \|\nabla J(\boldsymbol{\theta}_k)\|^2
        \ge
        \frac{(J^* - J(\boldsymbol{\theta}_k))^2}{\kappa_{\mathrm{GD}}^2}.
    \end{align*}
    Taking expectations and applying Jensen's inequality yields:
    \begin{align*}
        \mathbb{E}[\|\nabla J(\boldsymbol{\theta}_k)\|^2]
        \ge
        \frac{1}{\kappa_{\mathrm{GD}}^2}
        \mathbb{E}[(J^* - J(\boldsymbol{\theta}_k))^2]
        \ge
        \frac{\Delta_k^2}{\kappa_{\mathrm{GD}}^2}.
    \end{align*}
    Substituting this bound gives the quadratic recurrence:
    \begin{align}
        \Delta_{k+1}
        \le
        \Delta_k
        - \frac{\eta(1-\alpha)}{2\kappa_{\mathrm{GD}}^2}\Delta_k^2
        + \frac{L\eta^2}{2}(2\alpha^2 B^2 + \sigma_{\text{DG}}^2).
        \label{eq:gd_gap_recursion}
    \end{align}

    To achieve $\Delta_T \le \epsilon$, choose the step size such that:
    \begin{align*}
        \eta
        \le
        \min\left\{
        \frac{1}{2L(1-\alpha)},
        \frac{(1-\alpha)\epsilon^2}
        {2L\kappa_{\mathrm{GD}}^2(2\alpha^2 B^2 + \sigma_{\text{DG}}^2)},
        \frac{\kappa_{\mathrm{GD}}^2}{(1-\alpha)\epsilon}
        \right\}.
    \end{align*}
    Whenever $\Delta_k \ge \epsilon$, the second step-size bound ensures:
    \begin{align*}
        \frac{L\eta^2}{2}(2\alpha^2 B^2 + \sigma_{\text{DG}}^2)
        \le
        \frac{\eta(1-\alpha)}{4\kappa_{\mathrm{GD}}^2}\Delta_k^2.
    \end{align*}
    Equation~\ref{eq:gd_gap_recursion} therefore gives:
    \begin{align*}
        \Delta_{k+1}
        \le
        \Delta_k
        - \frac{\eta(1-\alpha)}{4\kappa_{\mathrm{GD}}^2}\Delta_k^2.
    \end{align*}
    If $\Delta_{k+1}>0$, taking reciprocals and using $\Delta_{k+1}\le\Delta_k$ gives:
    \begin{align*}
        \frac{1}{\Delta_{k+1}}-\frac{1}{\Delta_k}
        =
        \frac{\Delta_k-\Delta_{k+1}}{\Delta_{k+1}\Delta_k}
        \ge
        \frac{\eta(1-\alpha)}{4\kappa_{\mathrm{GD}}^2}.
    \end{align*}
    Summing this inequality until the optimality gap reaches $\epsilon$ yields:
    \begin{align*}
        \frac{1}{\Delta_T}
        \ge
        \frac{1}{\Delta_0}
        +
        \frac{T\eta(1-\alpha)}{4\kappa_{\mathrm{GD}}^2}.
    \end{align*}
    Hence, for $\Delta_0>\epsilon$, it suffices to choose:
    \begin{align*}
        T
        \ge
        \frac{4\kappa_{\mathrm{GD}}^2}{\eta(1-\alpha)}
        \left(\frac{1}{\epsilon}-\frac{1}{\Delta_0}\right).
    \end{align*}
    If $\Delta_0\le\epsilon$, the initial policy already satisfies the target accuracy.
    The third step-size bound ensures that the right-hand side of Equation~\ref{eq:gd_gap_recursion} is non-decreasing in $\Delta_k$ over $[0,\epsilon]$. Together with the second step-size bound, for any $\Delta_k\le\epsilon$,
    \begin{align*}
        \Delta_{k+1}
        &\le
        \epsilon
        - \frac{\eta(1-\alpha)}{2\kappa_{\mathrm{GD}}^2}\epsilon^2
        + \frac{L\eta^2}{2}(2\alpha^2 B^2 + \sigma_{\text{DG}}^2) \\
        &\le \epsilon.
    \end{align*}
    Thus, once the optimality gap reaches $\epsilon$, it remains within $[0,\epsilon]$, and the sufficient iteration bound applies to every subsequent iterate.

    In the asymptotic accuracy regime, the accuracy-dependent step-size bound is active. Setting
    \begin{align*}
        \eta
        =
        \frac{(1-\alpha)\epsilon^2}
        {2L\kappa_{\mathrm{GD}}^2(2\alpha^2 B^2 + \sigma_{\text{DG}}^2)}
    \end{align*}
    and substituting it into the sufficient iteration bound gives:
    \begin{align*}
        T
        \ge
        \frac{8L\kappa_{\mathrm{GD}}^4(2\alpha^2 B^2 + \sigma_{\text{DG}}^2)}
        {(1-\alpha)^2\epsilon^3}
        \left(1-\frac{\epsilon}{\Delta_0}\right).
    \end{align*}
    The quantities $B$, fixed $\alpha$, and $\sigma_{\text{DG}}^2 = \mathcal{O}(1)$ (Theorem~\ref{thm:variance}) are independent of $N$. When $L$ and $\kappa_{\mathrm{GD}}$ are also independent of $N$, the required number of iterations simplifies to:
    \begin{align*}
        T
        =
        \mathcal{O}\left(\frac{L\kappa_{\mathrm{GD}}^4}{\epsilon^3}\right),
    \end{align*}
    independent of $N$, completing the proof.
\end{proof}

\section{One-Step Congestion Game}
\label{app:env_congestion}

We introduce a one-step congestion game whose exact policy gradient can be computed in closed form, and we use it to measure how the gradient variance of the standard PG and guidance estimators scales with the number of agents $N$ (Section~\ref{subsec:alpha} and Figure~\ref{fig:alpha_sensitivity}(a)).

\subsection{Environment}

\paragraph{Game.}
The game has $N$ agents and $N$ resources. Each agent $i$ assigns one unit of demand to one of the $N$ resources, so its action is $a^i \in \{1, \dots, N\}$. Resource $n$ has capacity $c_n$ drawn uniformly from $[0.5, 1.5]$, and its load $x_n$ is the number of demand units assigned to it. The load vector $\boldsymbol{x} = (x_1, \dots, x_N)$ is the system state. All agents share the reward
\begin{equation*}
    r(\boldsymbol{x}) = -\sum_{n=1}^{N} \frac{x_n^2}{c_n},
\end{equation*}
which penalizes overloaded low-capacity resources most heavily. Since the game has a single step, the return equals $r(\boldsymbol{x})$, and the cooperative objective is $J(\boldsymbol{\theta}) = \mathbb{E}_{\boldsymbol{\pi}}[r(\boldsymbol{x})]$.

\paragraph{Reference state.}
We use the capacity-proportional balanced load $\tilde{x}_n = c_n N / \sum_{m=1}^{N} c_m$ as the reference state, which minimizes the congestion cost when loads may take fractional values. This reference depends only on the capacities, and since the total load $\sum_n x_n = N$ is conserved, the Cauchy-Schwarz inequality shows that moving $\boldsymbol{x}$ toward $\tilde{\boldsymbol{x}}$ improves the reward, as required by Assumptions~\ref{ass:exogeneity} and~\ref{ass:alignment}. Each load $x_n$ is a sum of independent Bernoulli variables whose mean stays below $2.1$ in every measured instance, so its one-step increment has a second moment bounded independently of $N$, as required by Assumption~\ref{ass:bounded_system_state}.

\subsection{Gradient Variance Measurement}
\label{app:variance_toy}

Because the policy gradient of this game has a closed form, we compute each estimator's variance around its exact mean, which is not available in the cloud scheduling and power dispatch environments.

\paragraph{Policies and estimators.}
Each agent uses an independent softmax policy over its $N$ options, with logits $\theta^i \in \mathbb{R}^N$ drawn from the standard normal distribution and then held fixed. We compare two estimators for agent $i$. The standard PG estimator is $\hat{g}_J^i = (r - J(\boldsymbol{\theta}))\, \nabla_{\theta^i} \log \pi^i(a^i)$. Subtracting the exact value $J(\boldsymbol{\theta})$ removes the mean return from this estimator and corresponds to an exact value baseline. The guidance estimator is $\hat{g}_{\mathcal{G}}^i = (x_{a^i} - \tilde{x}_{a^i})\, \nabla_{\theta^i} \log \pi^i(a^i)$. Its coefficient follows from the unit influence vector $\boldsymbol{z}^i = \boldsymbol{e}_{a^i}$, as derived in Appendix~\ref{app:influence_derivation}.

\paragraph{Exact gradients and variance measurement.}
Each load $x_n$ is a sum of independent Bernoulli variables, so its mean $\mathbb{E}[x_n]$ and second moment $\mathbb{E}[x_n^2]$ have closed forms. Let $p^i_n$ denote the probability that agent $i$ selects resource $n$. These two quantities are
\begin{equation*}
    \mathbb{E}[x_n] = \sum_{i=1}^{N} p^i_n, \qquad
    \mathbb{E}[x_n^2] = \sum_{i=1}^{N} p^i_n (1 - p^i_n) + \Big( \sum_{i=1}^{N} p^i_n \Big)^2.
\end{equation*}
They give closed forms of the objective $J(\boldsymbol{\theta}) = -\sum_{n} \mathbb{E}[x_n^2] / c_n$ and the expected deviation $\mathcal{G}(\boldsymbol{\theta}) = \frac{1}{2} \sum_{n} \mathbb{E}[(x_n - \tilde{x}_n)^2]$. The chain rule through each softmax then gives the exact gradients $\nabla_{\theta^i} J$ and $\nabla_{\theta^i} \mathcal{G}$. We then sample $M = 4000$ independent joint actions under the fixed policy and compute, for each sample $m$ and agent $i$, the estimates $\hat{g}_J^{i,(m)}$ and $\hat{g}_{\mathcal{G}}^{i,(m)}$. The per-agent variance of the standard PG estimator is the squared deviation from the exact gradient, averaged over samples and agents,
\begin{equation*}
    \sigma_J^2 = \frac{1}{N} \sum_{i=1}^{N} \frac{1}{M} \sum_{m=1}^{M} \big\| \hat{g}_J^{i,(m)} - \nabla_{\theta^i} J \big\|^2,
\end{equation*}
and the guidance variance $\sigma_{\mathcal{G}}^2$ follows by replacing $\hat{g}_J^{i,(m)}$ and $\nabla_{\theta^i} J$ with $\hat{g}_{\mathcal{G}}^{i,(m)}$ and $\nabla_{\theta^i} \mathcal{G}$. Each squared norm sums over all components of $\theta^i$, so $\sigma_J^2$ and $\sigma_{\mathcal{G}}^2$ are traces of the per-agent covariance matrices. We repeat every scale over 5 random draws of the capacities and logits. As a check, the sample means of $\hat{g}_J^i$ and $\hat{g}_{\mathcal{G}}^i$ match the exact gradients within sampling error at every scale, which confirms both the closed forms and the unbiasedness of the two estimators.

\section{Cloud Scheduling Environment}
\label{app:env_spec}

\subsection{Physical Environment}

\paragraph{Cluster Configuration.}
The cluster consists of $N$ servers drawn from 12 heterogeneous instance types spanning three hardware generations, with configurations based on AWS instance types. Table~\ref{tab:server_types} details the full configuration.

\begin{table}[h]
    \centering
    \caption{Server instance types used in the simulation, based on AWS instance configurations.}
    \label{tab:server_types}
    \begin{tabular}{llccccc}
        \toprule
        \textbf{Gen.} & \textbf{Instance Type} & \textbf{vCPUs} & \textbf{Memory (GB)} & \textbf{$\eta_{\text{CPU}}$} & \textbf{$\eta_{\text{Mem}}$} & \textbf{Weight} \\
        \midrule
        \multirow{3}{*}{Gen-1} & m4.8xlarge & 32 & 128 & 0.70 & 0.68 & 10\% \\
        & t3.2xlarge & 16 & 64 & 0.72 & 0.70 & 3\% \\
        & t3a.2xlarge & 16 & 64 & 0.75 & 0.72 & 2\% \\
        \midrule
        \multirow{6}{*}{Gen-2} & m5.8xlarge & 32 & 128 & 0.95 & 0.93 & 20\% \\
        & c5.18xlarge & 64 & 128 & 0.98 & 0.93 & 10\% \\
        & c5.12xlarge & 48 & 96 & 0.97 & 0.93 & 8\% \\
        & r5.8xlarge & 32 & 256 & 0.87 & 0.92 & 9\% \\
        & r5.6xlarge & 24 & 192 & 0.84 & 0.93 & 7\% \\
        & m5n.12xlarge & 48 & 192 & 0.96 & 0.92 & 6\% \\
        \midrule
        \multirow{3}{*}{Gen-3} & m6i.8xlarge & 32 & 128 & 1.06 & 0.95 & 15\% \\
        & c5n.18xlarge & 64 & 256 & 1.08 & 0.95 & 8\% \\
        & c6i.32xlarge & 96 & 384 & 1.08 & 0.96 & 2\% \\
        \bottomrule
    \end{tabular}
\end{table}

The efficiency parameters $\eta_{\text{CPU}}$ and $\eta_{\text{Mem}}$ model performance-per-watt differences across hardware generations. Gen-1 instances (circa 2015) have $\eta_{\text{CPU}} \in [0.70, 0.75]$, Gen-2 (circa 2017) have $\eta_{\text{CPU}} \in [0.84, 0.98]$, and Gen-3 (circa 2020+) have $\eta_{\text{CPU}} \in [1.06, 1.08]$. Job completion time scales as $\text{duration} / \eta_{\text{CPU}}$, making high-efficiency servers more productive. Server types are sampled according to the specified weights and randomly shuffled across indices at scenario generation to prevent agents from exploiting positional patterns. The CPU and memory capacities of each server are further scaled by independent factors drawn from $U(0.95, 1.05)$.

\paragraph{Workload Generation.}
We employ a \textit{bimodal task distribution} reflecting real-world workload diversity:
\begin{itemize}
    \item \textbf{CPU-intensive tasks} ($\sim$60\% of arrivals):
    CPU requirement drawn from Pareto($\alpha = 1.7$, $x_m = 0.6$), clipped to $[0.5, 20]$ cores;
    memory requirement $m = c \cdot U(1.5, 3.0)$, clipped to $[0.5, 64]$ GB (low memory-to-CPU ratio).

    \item \textbf{Memory-intensive tasks} ($\sim$40\% of arrivals):
    CPU requirement drawn from Pareto($\alpha = 2.2$, $x_m = 0.4$), clipped to $[0.2, 8]$ cores;
    memory requirement $m = c \cdot U(6, 12)$, clipped to $[1, 128]$ GB (high memory-to-CPU ratio).
\end{itemize}
All tasks share the same duration distribution: Log-normal($\mu = 3.0$, $\sigma = 0.5$), clipped to $[5, 150]$ timesteps, yielding mean duration $\sim$20 timesteps with high variance.

The Pareto shape parameter $\alpha \in [1.7, 2.2]$ is calibrated to match the heavy-tailed resource demand distributions observed in the Google cluster trace~\citep{reiss2011google}, where a small fraction of jobs consume disproportionate resources. The log-normal duration distribution is consistent with empirical observations from production clusters~\citep{reiss2012heterogeneity}. The bimodal CPU-to-memory ratio reflects the coexistence of compute-intensive (e.g., batch analytics) and memory-intensive (e.g., in-memory databases) workloads in modern data centers~\citep{delimitrou2014quasar}.

\paragraph{Arrival Process.}
Jobs arrive according to a non-homogeneous Poisson process with time-varying intensity:
\begin{equation*}
    \lambda(t) = \bar{\lambda} \cdot \max\left(0.1,\; 1 + 0.3\sin\left(\frac{2\pi t}{1000}\right) + \epsilon_t\right), \quad \epsilon_t \sim \mathcal{N}(0, 0.1^2)
\end{equation*}
where the base rate $\bar{\lambda}$ is calibrated via a capacity-aware formula to achieve target utilization $\rho \in [0.80, 0.85]$ across all scales. The sinusoidal term introduces a 1000-step tidal period, capturing the diurnal traffic patterns observed in production systems~\citep{cortez2017resource}, while $\epsilon_t$ adds stochastic perturbation. Figure~\ref{fig:environment} illustrates the bimodal workloads, server heterogeneity, and non-stationary arrivals described above.

\begin{figure}[!ht]
    \centering
    \includegraphics[width=\linewidth]{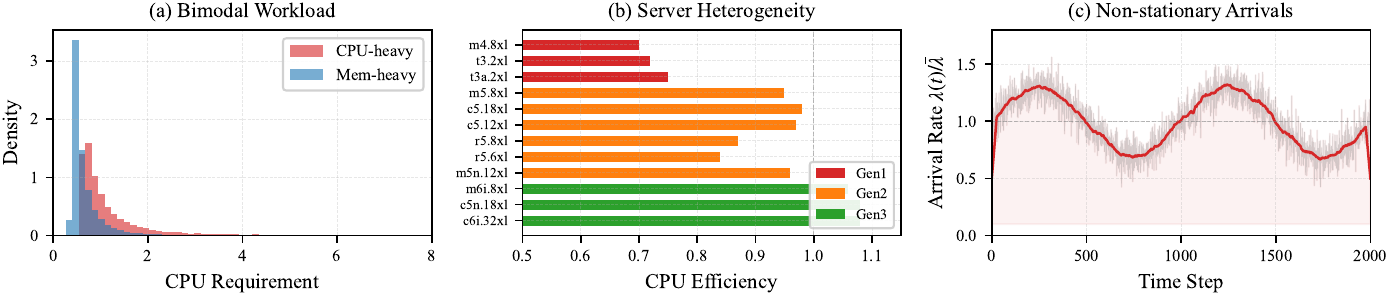}
    \caption{Environment characteristics. (a)~Bimodal workloads with heavy-tailed demands.
    (b)~Server heterogeneity across hardware generations.
    (c)~Non-stationary arrivals with tidal fluctuation.}
    \label{fig:environment}
\end{figure}

\subsection{Agent Interface}

At each timestep, newly arrived jobs are assigned to dedicated agents that select placement actions concurrently. We set the number of job-dispatch agents equal to the number of servers so that newly arrived jobs can be dispatched without accumulating at the assignment stage. Under this convention, $N$ denotes both the server scale and number of job-dispatch agents. We evaluate configurations from $N{=}2$ to $N{=}1500$.

\paragraph{Scale Configurations.}
Table~\ref{tab:app_scale} summarizes the experimental configurations. Each scale is defined by the number of servers $N$, the number of action-space clusters $K$, the number of servers per cluster $S = N/K$, and the target cluster utilization $\rho$ used to calibrate the arrival rate $\bar{\lambda}$ (see Arrival Process above). Each episode consists of 3000 timesteps, with job arrivals during the first 2000 and the remaining 1000 reserved as a drain period.

\begin{table}[h]
    \centering
    \caption{Scale configurations used in the experimental evaluation.}
    \label{tab:app_scale}
    \resizebox{\textwidth}{!}{%
    \begin{tabular}{ccccc}
        \toprule
        \textbf{Servers} ($N$) & \textbf{Clusters} ($K$) & \textbf{Servers/Cluster} ($S$) & \textbf{Episode Length} & \textbf{Target Util.} ($\rho$) \\
        \midrule
        2   & 2   & 1 & 3000 & $[0.80, 0.85]$ \\
        5   & 5   & 1 & 3000 & $[0.80, 0.85]$ \\
        10  & 10  & 1 & 3000 & $[0.80, 0.85]$ \\
        20  & 20  & 1 & 3000 & $[0.80, 0.85]$ \\
        50  & 25  & 2 & 3000 & $[0.80, 0.85]$ \\
        100 & 25  & 4 & 3000 & $[0.80, 0.85]$ \\
        200 & 40  & 5 & 3000 & $[0.80, 0.85]$ \\
        500 & 50  & 10 & 3000 & $[0.80, 0.85]$ \\
        1000 & 50 & 20 & 3000 & $[0.80, 0.85]$ \\
        1500 & 75 & 20 & 3000 & $[0.80, 0.85]$ \\
        \bottomrule
    \end{tabular}
    }
\end{table}

\paragraph{Action Space.}
For scales with $N \leq 20$, each agent directly selects one of the $N$ servers ($K = N$, $S = 1$). For scales with $N \geq 50$, the action space becomes prohibitively large, so we first sort servers by capacity $(c_k^{\text{CPU}}, c_k^{\text{Mem}})$ and then partition them into $K<N$ clusters of $S=N/K$ servers with similar profiles. Each agent selects a \textit{cluster} (action space of size $K$). Within the selected cluster, the simulator uses a deterministic Best-Fit rule to place the job on the server with the highest current utilization that still has sufficient free resources. If no server in the selected cluster can accommodate the job, it is queued on the least-loaded server in that cluster.

\paragraph{Observation Space.} Each agent receives an observation consisting of:
\begin{itemize}
    \item \textbf{Global server state} ($7N$ dims): For each server $k \in \{1, \ldots, N\}$, we include CPU utilization, memory utilization, local queue length (clipped at 50), CPU efficiency, memory efficiency, CPU capacity, and memory capacity.
    \item \textbf{All job information} ($2N_{\text{agents}}$ dims): CPU and memory requirements for jobs arrived at current timestep.
    \item \textbf{Job features} (2 dims): CPU and memory requirements of this agent's assigned job.
    \item \textbf{Temporal feature} (1 dim): Timestep $t / T_{\max}$.
    \item \textbf{Agent identifier} (1 dim): Agent index, used with a learnable embedding (see below).
\end{itemize}

\paragraph{Observation Compression.}
The raw observation provides the full server and job information needed for placement decisions. However, because it explicitly enumerates all servers and newly arrived jobs, its dimension grows quickly with $N$. We therefore replace the two high-dimensional components separately. Per-server states are represented by cluster summaries, and per-job information is represented by job summaries.
\begin{itemize}
    \item \textbf{Cluster summary.} Each cluster is represented by a 28-dimensional vector consisting of five groups of features.
    \begin{itemize}
        \item \textbf{CPU utilization.} Min, Q25, median, Q75, max, mean, and std.
        \item \textbf{Memory utilization.} Min, Q25, median, Q75, max, mean, and std.
        \item \textbf{CPU-memory joint distribution.} Correlation, fraction CPU-heavy, fraction memory-heavy, and fraction balanced.
        \item \textbf{Queue congestion.} Total, max, mean, std, and fraction non-empty.
        \item \textbf{Capacity and efficiency.} Free CPU fraction, free memory fraction, mean efficiency, cluster size, and overloaded fraction.
    \end{itemize}
    \item \textbf{Job summary.} Jobs arriving at the current timestep are represented by 6 aggregate job-demand features, including the mean, std, and max of CPU demand and the mean, std, and max of memory demand.
    \item \textbf{Retained features.} The agent's own job features, temporal feature, and agent identifier are kept unchanged.
\end{itemize}
The total observation dimension is therefore $28K + 10$, consisting of $28K$ cluster-summary features, 6 job-summary features, and 4 retained agent/time features. Since the compression is most beneficial once each cluster contains multiple servers, we apply it for $N \ge 200$ ($S=N/K \ge 5$), achieving up to $6.4\times$ reduction compared with the raw representation.

\paragraph{Agent ID Embedding.}
To enable parameter sharing across heterogeneous agents, we embed each agent's normalized index via a learnable \texttt{nn.Embedding} layer (dimension 16), which is concatenated with the remaining observation features before the first hidden layer. This allows a single shared network to learn agent-specific behaviors.

\subsection{Reward Function}
The shared reward signal is a weighted combination of two objectives:
\begin{equation*}
    r_t = -\left(w_1 \cdot \underbrace{\frac{Q_{\text{global}}(t) + \sum_{k=1}^{N} Q_k(t)}{N}}_{\text{Queue Penalty (Latency)}} + w_2 \cdot \underbrace{\frac{\sum_{k=1}^{N} \frac{L_k(t)}{\eta_k}}{C_{\text{total}}}}_{\text{Energy Penalty}}\right)
\end{equation*}
where $Q_{\text{global}}(t)$ is the global buffer size, $Q_k(t)$ is the local queue at server $k$, $L_k(t)$ is the workload (CPU + memory) at server $k$, $\eta_k$ is the server's efficiency, and $C_{\text{total}}$ is the total cluster capacity. We set $w_1 = 1$ and $w_2 = 20$ so that both penalties contribute comparably to the reward. Without the multiplier, the queue penalty typically ranges in $[0, 20]$, whereas the normalized energy ratio remains below $1$. Load-balancing (utilization variance) is handled entirely by the DG-PG guidance signal and is therefore excluded from the reward.

\subsection{Baseline Specifications}

\textbf{Best-Fit.} Best-Fit is a centralized heuristic that assigns each arriving job to the feasible server with the highest utilization, that is, the tightest fit in CPU and memory, and queues the job on the least-utilized server if no server can accommodate it. Jobs arriving in the same timestep are assigned sequentially, and each assignment updates the free resources seen by the next. For $N \geq 50$, Best-Fit acts through the same cluster-level action as the learned policies, submitting the cluster of its selected server.

\textbf{Random Dispatch.} For $N \leq 20$, Random Dispatch assigns each arriving job to a uniformly random server. For $N \geq 50$, it uses the same cluster-level action as DG-PG, selecting a cluster uniformly at random, and the job is then placed by the same within-cluster rule. Random Dispatch thus differs from DG-PG only in how the cluster is selected, so the gap between them is not attributable to the within-cluster placement rule.

\section{Power Dispatch Environment}
\label{app:env_power}

This appendix specifies the single area dynamic economic dispatch environment, which has a continuous action space. The system contains $N$ heterogeneous thermal generators and a varying aggregate net load. At every control step, generator $i$ selects a normalized continuous action $a_{i,k}\in[-1,1]$, which changes its power command subject to its ramp and generation limits. The joint objective is to regulate system frequency and allocate generation economically. The physical model follows the single area load frequency control setting of \citet{rozada2021deep}. Generator capacities, operating limits, and quadratic costs are constructed from five thermal unit records in the Reliability Test System-Grid Modernization Laboratory Consortium data \citep{barrows2020ieee}.

\subsection{Physical Environment}

Let $p_i(t)$ and $z_i(t)$ denote generator $i$'s actual power and command, respectively, and let $\Delta f(t)$ denote the deviation from the nominal frequency. At control step $k$, the environment applies
\begin{equation}
z_{i,k}^{+}
=
\operatorname{clip}\!\left(
z_{i,k}+\overline u_i a_{i,k},
\underline p_i,
\overline p_i
\right),
\qquad a_{i,k}\in[-1,1],
\label{eq:energy_action}
\end{equation}
and holds $z_{i,k}^{+}$ fixed over the next control interval. Generator output and system frequency then evolve according to
\begin{equation}
T_i\dot p_i(t)=-p_i(t)+z_{i,k}^{+}-\frac{1}{R_i}\Delta f(t),
\qquad
M_f\dot{\Delta f}(t)=\sum_{i=1}^{N}p_i(t)-P_L(t)-D_f\Delta f(t).
\label{eq:energy_dynamics}
\end{equation}
The stochastic net load is $P_L(t)=P_L^0-\Delta P_W(t)$, clipped to the feasible aggregate generation interval, where the wind-power process follows \citet{rozada2021deep}:
\begin{equation}
\dot{\Delta P}_W(t)=-0.002\Delta P_W(t)+0.01\Delta v(t),
\qquad
\mathrm{d}\Delta v(t)=-0.5\Delta v(t)\,\mathrm{d}t-0.4s\,\mathrm{d}W_t,
\label{eq:energy_load}
\end{equation}
with scale factor $s=N/10$. The base load, total capacity, and aggregate action limit are $3s$, $6s$, and $0.1s$ per unit, respectively. Aggregate inertia and damping scale by $s$, while aggregate droop scales by $1/s$. This preserves the normalized physical dynamics as $N$ increases. Parameters within each generator cluster receive fixed $\pm2\%$ perturbations.

\subsection{Agent Interface}

Each episode represents 40 minutes. Actions are updated every 2 seconds, giving 1200 control steps, while the differential equations are integrated every 0.05 seconds. The 2 second control interval and 40 minute evaluation window follow the Automatic Generation Control scan rate and regulation test duration in Pennsylvania-New Jersey-Maryland Interconnection Manual 12 \citep{pjm2026manual12}. 

Each actor receives frequency deviation, net load, aggregate actual generation, aggregate command, and time, together with its own ten state and parameter features. As in cloud scheduling (Appendix~\ref{app:env_spec}), the global view contains the complete normalized generator state for $N<200$ and is compressed into 28 statistics per cluster for $N\geq200$. These statistics describe output, commands, marginal costs, reserves, capacity, and boundary activity. Also as in cloud scheduling, a 16-dimensional learned embedding represents the generator identifier, and agents share the actor parameters while independently outputting their actions.

\subsection{Reward and Analytical Reference}

Generator $i$ has quadratic operating cost $C_i(p_i)=a_i p_i^2+b_i p_i+c_i$. Let $G_{k+1}=\sum_i p_{i,k+1}$, and let $C^\star(G_{k+1})$ be the minimum generation cost among all feasible allocations with the same aggregate generation $G_{k+1}$. Using a frequency scale of $0.05\,\mathrm{Hz}$, the shared environment reward is
\begin{equation}
r_k
=
-1000\Delta t\left[
\left(\frac{\Delta f_{k+1}}{0.05\ \mathrm{Hz}}\right)^2
+
\frac{\sum_i C_i(p_{i,k+1})-C^\star(G_{k+1})}
{C^\star(G_{k+1})}
\right],
\qquad \Delta t=2\ \mathrm{s}.
\label{eq:energy_reward}
\end{equation}
The frequency term measures the balance between supply and demand. The second term measures economic inefficiency at the realized aggregate generation.

For the current net load $P_{L,k}$, the unique steady state reference aligned with the reward is the bounded economic dispatch solution \citep{wood2013power}:
\begin{equation}
\begin{aligned}
p^{\mathrm{ref}}(P_{L,k})
\in\arg\min_p\quad
&\sum_{i=1}^{N}C_i(p_i)\\
\mathrm{s.t.}\quad
&\sum_{i=1}^{N}p_i=P_{L,k},\\
&\underline p_i\leq p_i\leq\overline p_i.
\end{aligned}
\label{eq:energy_reference}
\end{equation}
At this state, $\Delta f^{\mathrm{ref}}=0$ and $p^{\mathrm{ref}}=z^{\mathrm{ref}}$, so both terms in Equation~\ref{eq:energy_reward} attain zero. We solve Equation~\ref{eq:energy_reference} by bounded water filling. The resulting $p^{\mathrm{ref}}$ specifies a target state only. It is used during training to compute the change in each generator's normalized command distance,
\begin{equation}
c_{i,k}
=
\frac{1}{2}\left(
\frac{z_{i,k}^{+}-p_{i,k}^{\mathrm{ref}}}{\overline p_i}
\right)^2
-
\frac{1}{2}\left(
\frac{z_{i,k}-p_{i,k}^{\mathrm{ref}}}{\overline p_i}
\right)^2,
\label{eq:energy_guidance}
\end{equation}
and is unavailable to the learned policy during evaluation.

\subsection{Baseline Specifications}

\textbf{Centralized greedy.} At each control step, the controller observes the net load and the states of all generators, and predicts the frequency and generator outputs at the end of the next control interval by integrating Equation~\ref{eq:energy_dynamics}. It then visits generators in descending order of capacity and selects each action $a_{i,k}$ by a bounded one-dimensional search that minimizes the predicted one-step loss in Equation~\ref{eq:energy_reward}, given the actions already selected.

\textbf{AGC.} We implement integral AGC with economic participation factors \citep{wood2013power}. At each control step, generator $i$ changes its command by $-\beta_i K_I \Delta f_k \Delta t$, where $K_I$ is the integral gain and $\beta_i$ is the participation factor of generator $i$. The participation factor $\beta_i$ is proportional to $1/(2a_i)$ for generators strictly within their limits at the initial economic dispatch and zero otherwise, normalized so that $\sum_i \beta_i = 1$. A generator at its limit receives no command change that would push it further beyond that limit. For each scale, we select $K_I$ from $\{0.5, 1, 2\} \times 1.75s$ by the mean reward over 5 validation seeds.




\section{Training Configuration}
\label{app:hyperparams}

\subsection{Training and Evaluation Protocol}

For training, each run uses an independent random seed for model initialization and stochastic simulation. In cloud scheduling, we run 3 seeds for each $\alpha$ and scale in the $\alpha$-selection experiment, 5 seeds for each method and scale in the controlled comparison, and 5 seeds for $N \leq 100$ and 3 seeds for $N \geq 200$ in the scalability experiment. Power dispatch uses the same seed setting as the cloud scalability experiment. For evaluation, checkpoints are saved every 5 episodes and tested on the same 10 held-out test seeds, and we report the checkpoint with the highest mean test reward. In cloud scheduling, we test DG-PG checkpoints up to Episode 30 and MAPPO and IPPO checkpoints over their full training. DG-PG, MAPPO, and IPPO report mean $\pm$ standard deviation across training seeds, while Best-Fit and Random Dispatch require no training and are evaluated on the same test seeds.

\subsection{Training Settings}

DG-PG, MAPPO, and IPPO are all trained with PPO using parameter sharing across agents. Across all learned methods, we set the discount factor to $\gamma=0.99$ and the GAE parameter to $\lambda=0.95$, and use Adam with separated actor-critic updates, independent gradient clipping, return normalization before critic updates, and advantage normalization before policy updates.

However, MAPPO and IPPO with neural-network actor and critic models failed to learn competitive policies despite extensive tuning (Appendix~\ref{app:nn_mappo_failure}). The controlled comparison therefore uses linear actor and critic models, together with stronger stabilization settings, for all three methods. This keeps the actor and critic function classes matched within the comparison, while the guidance-weight and scalability experiments use the neural-network DG-PG configuration.

\subsubsection{Linear Models (Controlled Comparison)}

For the controlled comparison in cloud scheduling at $N \in \{2, 5, 10\}$, all three methods use the same linear actor $\pi(a|s)=\mathrm{softmax}(Ws+b)$ and linear critic $V(s)=w^\top s+b$. They also share the same core training parameters. Table~\ref{tab:linear_hyper} lists the full configuration.

\begin{table}[h]
    \centering
    \caption{Hyperparameters for controlled comparison (Section~\ref{subsec:main_results}). }
    \label{tab:linear_hyper}
    \begin{tabular}{ll}
        \toprule
        \textbf{Parameter} & \textbf{Value} \\
        \midrule
        Model class & Linear (single layer, no hidden units) \\
        Learning rate & $1 \times 10^{-3}$ (linear decay to $1\%$) \\
        PPO clipping $\epsilon$ & $0.2$ \\
        Critic epochs / Actor epochs & 20 / 3 (split-epoch) \\
        Mini-batch size & 10{,}000 \\
        Gradient clipping & max norm $10.0$ \\
        Huber loss $\delta$ & $10.0$ \\
        Parallel environments & 4 \\
        Rollouts per episode & 24 \\
        \midrule
        DG-PG entropy coefficient & $0.01$ (exponential decay to $10^{-4}$) \\
        MAPPO / IPPO entropy coefficient & $0.005$ (fixed, no decay) \\
        DG-PG training episodes & 200 \\
        MAPPO / IPPO training episodes & 500 \\
        DG-PG $\alpha$ schedule & dynamic schedule ($0.9 \to 0.2$) \\
        \bottomrule
    \end{tabular}
\end{table}

Several shared parameters were set to stabilize MAPPO and IPPO, based on the tuning documented in Appendix~\ref{app:nn_mappo_failure}:
\begin{itemize}
    \item \textbf{Split-epoch training} (critic 20 epochs, actor 3 epochs): MAPPO and IPPO require a well-trained critic before the actor receives meaningful advantage signals. A single linear layer needs many gradient steps to fit the value function. Without this asymmetric schedule, advantage estimates are too noisy for stable policy updates.
    \item \textbf{Large mini-batch} (10{,}000) and \textbf{high rollout count} (24): With smaller batches (e.g., 256--1{,}024), MAPPO and IPPO exhibit large reward oscillations and do not converge. Multi-agent policy gradients in this environment have high variance due to heterogeneous servers, bimodal workloads, and non-stationary arrivals, requiring large batches to average over noise.
\end{itemize}
These settings are based on the MAPPO implementation guidelines in \citet{yu2022surprising}, with additional stabilization needed for our scheduling environment. DG-PG uses the same linear architecture and the shared training settings for fair comparison.

The remaining differences also favor MAPPO/IPPO:
\begin{itemize}
    \item \textbf{Entropy}: MAPPO/IPPO use a fixed coefficient ($0.005$) rather than decay, because entropy decay caused premature policy collapse in our tuning. The fixed value was selected as the best-performing setting across grid search.
    \item \textbf{Training length}: MAPPO/IPPO are trained for 500 episodes ($2.5\times$ longer than DG-PG's 200), providing additional convergence time.
\end{itemize}

\subsubsection{Neural Network Models}

For the DG-PG experiments in Sections~\ref{subsec:alpha} and~\ref{subsec:scalability}, which cover cloud scheduling and power dispatch, we use neural-network actor and critic models rather than the linear models used in the controlled comparison. Both models are two-layer feedforward networks with Tanh activations.

\paragraph{Hyperparameters.}
The following hyperparameters apply to cloud scheduling. Power dispatch differences are listed at the end of this subsection.
\begin{itemize}
    \item \textbf{Learning rate}: $1 \times 10^{-4}$ (exponential decay)
    \item \textbf{Entropy coefficient}: $0.02$ (exponential decay to $10^{-4}$)
    \item \textbf{PPO update epochs}: $K = 4$ (unified for actor and critic)
    \item \textbf{Gradient clipping}: max norm $0.5$
    \item \textbf{Guidance signal}: See Appendix \ref{app:verification_cloud}
    \item \textbf{Guidance weight $\alpha$ in Section~\ref{subsec:alpha}}: fixed $\alpha \in \{0.2,0.4,0.6,0.8\}$
    \item \textbf{Guidance weight $\alpha$ in Section~\ref{subsec:scalability}}: dynamic schedule for $N \in \{5,10,20,50\}$ and fixed $\alpha = 0.9$ for $N \ge 100$
    \item \textbf{Guidance normalization}: running guidance statistics are updated with momentum $0.99$, and the normalized guidance signal is clipped to $[-3, 3]$
\end{itemize}

\paragraph{Scale-Specific Settings.}
For the scalability experiments, we use scale-specific settings according to the input size and training cost, as summarized in Table~\ref{tab:app_hyper}.

\begin{table}[h]
    \centering
    \caption{Scale-specific hyperparameters.}
    \label{tab:app_hyper}
    \resizebox{\textwidth}{!}{%
    \begin{tabular}{ccccccc}
        \toprule
        \textbf{Servers} ($N$) & \textbf{Hidden Dimension} & \textbf{Mini-batch} & \textbf{Rollouts} & \textbf{Parallel Environments} & \textbf{PPO Clip} ($\epsilon$) & \textbf{Active Samples} \\
        \midrule
        10   & 128  & 512  & 12 & 4 & 0.2 & $\sim 120{,}000$ \\
        20   & 128  & 512  & 12 & 4 & 0.2 & $\sim 240{,}000$ \\
        50   & 256  & 1024 & 8  & 4 & 0.4 & $\sim 400{,}000$ \\
        100  & 512  & 1024 & 6  & 2 & 0.4 & $440{,}000$ (cap)\\
        200  & 1024 & 1024 & 2  & 1 & 0.4 & $440{,}000$ (cap) \\
        500  & 1024 & 1024 & 1  & 1 & 0.4 & $440{,}000$ (cap) \\
        1000 & 1024 & 1024 & 1  & 1 & 0.4 & $440{,}000$ (cap) \\
        1500 & 1024 & 1024 & 1  & 1 & 0.4 & $440{,}000$ (cap) \\
        \bottomrule
    \end{tabular}
    }
\end{table}
As the observation and action spaces grow, we increase the hidden dimension of the actor and critic. At smaller scales, each rollout contains fewer active-agent transitions, so we use smaller mini-batches, more rollouts per episode, and more parallel environment workers to improve trajectory diversity. For $N \ge 50$, we use a wider PPO clipping ratio $\epsilon=0.4$, because DG-PG tolerated more aggressive policy updates and smaller clipping ratios slowed early learning. At larger scales, the number of active samples collected per episode grows rapidly, so we cap the retained active samples used for PPO updates at approximately $440{,}000$.

\paragraph{Training Stability.}
Since the guidance signal stabilizes DG-PG by reducing gradient variance (Theorem~\ref{thm:variance}), DG-PG does not require the stabilization mechanisms used for MAPPO and IPPO, such as split-epoch training, larger mini-batches, or more rollouts. Moreover, our experiments show that DG-PG can go beyond conservative stabilization settings and tolerate more aggressive PPO updates, such as a wider clipping ratio.

\paragraph{Power Dispatch.}
Power dispatch follows the cloud scheduling configuration above except in four aspects. First, because actions are continuous, the shared actor outputs a Gaussian policy followed by a hyperbolic tangent transformation, and evaluation uses the deterministic action $\tanh(\mu_i)$. A shared distributed critic produces an individual value estimate for each agent from the actor features. Second, we use a fixed entropy coefficient of $0.01$ and a fixed guidance weight $\alpha=0.9$ at all scales, and train for 50 episodes. Third, the guidance coefficient is normalized by its running root mean square without mean subtraction before clipping to $[-3,3]$. Fourth, rewards are divided by 1000 during optimization, while all reported rewards use the scale in Equation~\ref{eq:energy_reward}.

\section{Baseline Failure Analysis}
\label{app:nn_mappo_failure}

MAPPO and IPPO fail to learn competitive policies with neural-network actor and critic models due to a fundamental \textit{actor-critic update frequency dilemma}: the actor and critic have conflicting requirements for gradient update frequency, and no configuration resolves both simultaneously. We verified this across $\sim$20 training runs at $N{=}5$ and $N{=}2$, varying batch size, update epochs, entropy coefficient, value normalization, gradient clipping, and model architecture. Table~\ref{tab:mappo_runs} lists representative runs at $N{=}5$ with a two-layer feedforward network (128 hidden units). Each representative run uses approximately $200{,}000$ training samples per episode.

\begin{table}[h]
\centering
\caption{Representative MAPPO runs with neural-network actor and critic models at $N{=}5$. }
\label{tab:mappo_runs}
\resizebox{\textwidth}{!}{%
\setlength{\tabcolsep}{3pt}
\begin{tabular}{clcccl}
\toprule
\textbf{Run} & \textbf{Key Change} & \textbf{Mini-batch} & \textbf{PPO Epochs} & \textbf{$c_{\text{ent}}$} & \textbf{Observed Failure} \\
\midrule
0 & DG-PG defaults & 256 & 4 & 0.02 (decay) & Entropy collapse ($1.59 \to 0.26$) \\
1 & +ValueNorm & 256 & 10 & 0.01 & Entropy collapse ($1.60 \to 0.88$) \\
2 & $K$: $10 \to 4$ & 256 & 4 & 0.01 & Slow entropy collapse ($1.54 \to 1.36$) \\
3 & $c_{\text{ent}}$: $0.01 \to 0.03$ & 256 & 4 & 0.03 & Entropy bonus drowns advantage \\
4 & Align with MAPPO defaults & 50{,}000 & 10 & 0.01 & Critic too slow (explained variance $\leq 0.30$) \\
5 & Batch: $50\text{K} \to 20\text{K}$ & 20{,}000 & 15 & 0.01 & Critic still insufficient \\
6 & Batch: $20\text{K} \to 10\text{K}$ & 10{,}000 & 15 & 0.01 & No convergence \\
7 & +Potential-based shaping & 10{,}000 & 15 & 0.01 & Shaping signal negligible \\
\bottomrule
\end{tabular}%
}
\end{table}

\subsection{The Gradient Update Frequency Dilemma}

In MAPPO and IPPO training, the most sensitive quantity is the number of mini-batch gradient updates performed after each episode. When this number is large, the critic receives enough optimization steps, but the actor repeatedly updates on noisy advantages and quickly loses entropy. When this number is small, the actor remains stable, but the critic does not learn accurate value estimates. Table~\ref{tab:mappo_runs} summarizes these outcomes in three regimes:
\begin{itemize}
    \item \textbf{High update frequency} (Runs 0--3): The critic learns well (explained variance reaches 0.7), but the actor undergoes entropy collapse, converging to a near-deterministic, suboptimal action within 50--70 episodes, after which performance degrades irreversibly.
    \item \textbf{Low update frequency} (Runs 4--5): Entropy remains stable, but the critic receives too few updates to learn accurate value estimates (explained variance $\leq 0.30$), leaving the actor with uninformative advantages.
    \item \textbf{Intermediate update frequency} (Run 6): The update frequency avoids immediate entropy collapse and improves critic explained variance to 0.35--0.49, but the training curve oscillates for 1{,}000 episodes with no convergence trend.
\end{itemize}

\subsection{Attempts to Break the Gradient Update Frequency Dilemma}

We tested four structural modifications, none of which resolved the fundamental conflict:

\paragraph{ValueNorm~\citep{yu2022surprising} (Run 1).}
Value normalization improved critic fitting by increasing the critic explained variance from 0.0 to 0.7 under the high-update-frequency setting. However, the actor still suffered entropy collapse, showing that better value normalization alone does not resolve the update-frequency dilemma.

\paragraph{MAPPO implementation alignment (Run 4).}
We aligned the configuration with the MAPPO reference implementation from \citet{yu2022surprising}, including the mini-batch size ($256 \to 50{,}000$), value loss coefficient ($0.5 \to 1.0$), Huber loss parameter ($1.0 \to 10.0$), gradient clipping threshold ($0.5 \to 10.0$), entropy setting, and rollout count ($12 \to 24$). This eliminated entropy collapse by reducing gradient updates from 3{,}064 to 40 per episode, but made the critic too slow.

\paragraph{Potential-based reward shaping (Run 7).}
Potential-based reward shaping with $\Phi(s) = -\text{Var}(\text{utilization})$ produced negligible shaping signals. This indicates that the dominant bottleneck is gradient variance from cross-agent credit assignment, rather than temporal credit assignment.

\paragraph{Asynchronous actor-critic ($N{=}2$).}
Separate critic-then-actor optimization at the smallest scale produced near-zero KL divergence and 0\% PPO clipping. This suggests that the policy-gradient signal remains too noisy to drive meaningful policy updates even with only 2 agents.

\subsection{Transition to Linear Models}

These failures make neural-network MAPPO and IPPO unsuitable for a controlled comparison in this environment. We therefore use matched linear actor and critic models, $\pi(a|s) = \text{softmax}(Ws+b)$ and $V(s)=w^\top s+b$, for all three methods. The linear models reduce function-approximation capacity and make critic fitting more stable, which allows the split-epoch training scheme in Appendix~\ref{app:hyperparams} to be effective. Even under this simplified model class, MAPPO and IPPO still require additional stabilization, including split-epoch training, large mini-batches, and elevated rollout counts. This gives the baselines a more favorable training configuration and makes the comparison a conservative test of DG-PG's guidance signal.

These failures reflect not only the cross-agent variance scaling analyzed in Section~\ref{sec:preliminary}, but also the statistical difficulty of the cloud scheduling environment.
\begin{itemize}
    \item \textit{Server heterogeneity} ($\eta_{\text{CPU}} \in [0.70, 1.08]$): The best placement decisions depend strongly on the server mix, making value estimation more scenario-dependent.
    \item \textit{Non-stationary arrivals}: Poisson processes with time-varying rates make episodes structurally different, reducing the transferability of value estimates across episodes.
    \item \textit{Bimodal workloads}: Jobs arrive with two distinct resource profiles, increasing variability in returns and advantage estimates.
\end{itemize}
Together, these factors make both value estimation and policy learning difficult, because the policy must infer a highly state-dependent placement rule from noisy multi-agent advantage estimates. DG-PG addresses this difficulty by providing an analytical descent direction during training, rather than relying solely on end-to-end discovery from returns.

\section{Computational Resources}
\label{app:compute}

All experiments use CPU-only computation without GPU acceleration. Cloud scheduling runs on a single Apple M1 chip (8-core CPU, 16 GB unified memory). Power dispatch runs on an Apple M4 Max chip (16-core CPU, 64 GB unified memory) with up to 8 runs executed concurrently, so its timings include contention between concurrent runs.

In both environments, we define convergence as the first episode at which the training reward reaches within 10\% of the converged reward, the mean training reward over Episodes 31--50 of the same run. We compute this episode from each seed's unsmoothed training log and average it across training seeds, together with the per-episode training time. The convergence time is the product of the two. Table~\ref{tab:convergence_time} reports an average convergence episode of 5.3 in cloud scheduling and 7.3 in power dispatch over the 8 scales.

\begin{table}[h]
    \centering
    \caption{DG-PG convergence episode, convergence time, and per-episode training time (CPU-only).}
    \label{tab:convergence_time}
    \begin{tabular}{c ccc ccc}
        \toprule
        & \multicolumn{3}{c}{\textbf{Cloud scheduling}} & \multicolumn{3}{c}{\textbf{Power dispatch}} \\
        \cmidrule(lr){2-4} \cmidrule(lr){5-7}
        $N$ & \textbf{Episode} & \textbf{Time} & \textbf{Per-Episode} & \textbf{Episode} & \textbf{Time} & \textbf{Per-Episode} \\
        \midrule
        10   & 5.4 & 0.9 min  & 10.5 s  & 6.2 & 2.8 min  & 26.8 s \\
        20   & 5.8 & 1.8 min  & 18.9 s  & 8.2 & 4.4 min  & 31.9 s \\
        50   & 5.0 & 3.4 min  & 40.9 s  & 8.2 & 6.9 min  & 50.2 s \\
        100  & 4.0 & 9.6 min  & 2.4 min & 9.0 & 12.7 min & 1.4 min \\
        200  & 4.3 & 18.2 min & 4.2 min & 2.7 & 6.6 min  & 2.5 min \\
        500  & 5.7 & 19.8 min & 3.5 min & 7.7 & 17.6 min & 2.3 min \\
        1000 & 5.7 & 23.8 min & 4.2 min & 9.0 & 22.8 min & 2.5 min \\
        1500 & 6.3 & 35.5 min & 5.6 min & 7.7 & 22.9 min & 3.0 min \\
        \bottomrule
    \end{tabular}
\end{table}

Per-episode time is not monotonic in $N$ because the number of rollouts per episode decreases with scale (Table~\ref{tab:app_hyper}). For example, $N{=}500$ uses one rollout per episode, whereas $N{=}200$ uses two. Even at $N{=}1500$, DG-PG reaches the convergence threshold in about 35 minutes in cloud scheduling and 23 minutes in power dispatch.


\end{document}